\documentclass{article} 

\usepackage{amsthm, amsmath, amsfonts, amssymb, bm, mathtools, booktabs}
\usepackage{geometry}
\usepackage{algorithm} 
\usepackage{algorithmicx}
\usepackage{algpseudocode} 

\usepackage{caption} 
\usepackage{subcaption} 
\usepackage{nicefrac, graphicx, xcolor, thmtools, url}

\usepackage{gensymb}
\definecolor{darkpastelgreen}{rgb}{0.01, 0.75, 0.24}
\usepackage[hidelinks]{hyperref}
\usepackage[capitalise]{cleveref}

\hypersetup{colorlinks={true},linkcolor={blue},citecolor=darkpastelgreen}
\usepackage[square, numbers]{natbib}

\usepackage{amsmath,amsfonts,bm}

\def\eqref#1{equation~\ref{#1}}

\def\1{\bm{1}}

\def\vzero{{\bm{0}}}
\def\vone{{\bm{1}}}
\def\va{{\bm{a}}}
\def\vb{{\bm{b}}}
\def\vc{{\bm{c}}}

\def\vi{{\bm{i}}}
\def\vj{{\bm{j}}}

\def\vp{{\bm{p}}}

\def\vr{{\bm{r}}}
\def\vs{{\bm{s}}}

\def\vw{{\bm{w}}}

\def\mA{{\bm{A}}}

\DeclareMathAlphabet{\mathsfit}{\encodingdefault}{\sfdefault}{m}{sl}
\SetMathAlphabet{\mathsfit}{bold}{\encodingdefault}{\sfdefault}{bx}{n}

\newcommand{\R}{\mathbb{R}}

\theoremstyle{plain}
\newtheorem{theorem}{Theorem}[section]
\newtheorem{corollary}[theorem]{Corollary}
\newtheorem{claim}[theorem]{Claim}
\newtheorem{lemma}[theorem]{Lemma}
\newtheorem{proposition}[theorem]{Proposition}
\newtheorem{conjecture}[theorem]{Conjecture}

\newtheorem*{question*}{Question}

\theoremstyle{definition}
\newtheorem{Definition}[theorem]{Definition}
\newtheorem{Remark}[theorem]{Remark}
\newtheorem{example}[theorem]{Example}

\Crefname{claim}{Claim}{Claim}
\Crefname{conjecture}{Conjecture}{Conjecture}

\crefname{Definition}{Definition}{Definitions}
\crefname{@algorithm}{Algorithm}{Algorithms}

\DeclareMathOperator{\query}{\texttt{query}}
\DeclareMathOperator{\update}{\texttt{update}}
\DeclareMathOperator{\level}{{level}}
\DeclareMathOperator{\pred}{{pred}}
\DeclareMathOperator{\pt}{{par}}
\DeclareMathOperator{\sib}{{sib}}
\DeclareMathOperator{\buy}{\texttt{buy}}
\DeclareMathOperator{\cost}{\texttt{cost}}
\DeclareMathOperator{\price}{\texttt{price}}
\DeclareMathOperator{\tradef}{\texttt{forward\_trade}}
\DeclareMathOperator{\tradeb}{\texttt{backward\_trade}}

\DeclareMathOperator{\pend}{pend}
\DeclareMathOperator{\val}{val}

\newcommand{\defn}[1]{{\textbf{\textit{#1}}}}
\newcommand{\card}[1]{\lvert#1\rvert}
\newcommand{\M}{\mathcal{M}}
\newcommand{\U}{\mathcal{U}}
\newcommand{\V}{\mathcal{V}}
\newcommand{\A}{\mA}

\newcommand{\Cx}{\tilde{C}}
\newcommand{\veta}{\boldsymbol{\eta}}

\newcommand{\Parens}[1]{\left(#1\right)}
\newcommand{\bigParens}[1]{\bigl(#1\bigr)}
\newcommand{\BigParens}[1]{\Bigl(#1\Bigr)}

\def\vmu{{\bm{\mu}}}

\newcommand\fang[1]{\textcolor{olive}{[F:#1]}}
\newcommand\xw[1]{\textcolor{blue}{[xw: #1]}}
\renewcommand\fang[1]{}
\renewcommand\xw[1]{}

\title{Designing Automated Market Makers for Combinatorial Securities: A Geometric Viewpoint}

\author{Prommy Sultana Hossain\thanks{George Mason University, \texttt{phossai@gmu.edu}}
\and Xintong Wang\thanks{Rutgers University, \texttt{xintong.wang@rutgers.edu}}\and Fang-Yi Yu\thanks{George Mason University, \texttt{fangyiyu@gmu.edu}}}

\date{}
\begin{document}

\maketitle
\begin{abstract}

Designing automated market makers (AMMs) for prediction markets on combinatorial securities over large outcome spaces poses significant computational challenges.  
Prior research has studied combinatorial prediction markets on specific set systems (e.g., intervals, permutations), characterizing and addressing the design challenges by exploiting their respective security or outcome structures.  A comprehensive framework of AMMs design for prediction markets on \textit{arbitrary} set systems remains yet elusive.  In this paper, through establishing a novel connection between the design of AMMs for combinatorial prediction markets and the range query problem in computational geometry, we present a unified framework for both analyzing the computational complexity and designing efficient AMMs for combinatorial prediction markets. 

We first demonstrate the equivalence between price queries and trade updates under the popular combinatorial \emph{logarithmic market scoring rule} market (LMSR) and the range query and range update problem (RQRU).  
We show that combinatorial LMSR that allow efficient AMMs can be characterized by their securities' associated set system having a \emph{bounded} VC dimension.  
Specifically, we construct a partition-tree-based scheme to support price queries and trade updates in time sublinear to the number of outcomes, when the VC dimension is bounded, and show the non-existence of sublinear time AMM when the VC dimension is unbounded.  
We then generalize to AMMs for combinatorial prediction market with other scoring rules (quadratic and power scoring rules) by illustrating their connection to RQRU with different update rules and employing variations of the partition tree scheme.  
Finally, we show that the \emph{multi-resolution market design} can be naturally integrated into the partition-tree scheme.  
This facilitates the creation of submarkets with varying liquidity parameters and scoring rules without compromising computational efficiency or introducing arbitrage opportunities.  

Moreover, we introduce the combinatorial swap operation problem for automated market makers in decentralized finance and show that it can be efficiently reduced to range update problems.



\end{abstract}
\clearpage
\setcounter{page}{1}

\section{Introduction}
In a prediction market, traders buy and sell securities that pay out based on the outcomes of future events, with market prices reflecting predictions of those events.    
These markets have found applications across various domains, such as business for forecasting sales trends~\cite{bertsekas2015convex,DBLP:journals/mansci/SpannS03}, politics for predicting election results~\cite{hanson1999,berg2008results}, and entertainment for anticipating box office performance~\cite{pennock2001real}. 
Recently, prediction markets have been built on blockchain technologies (e.g., Augur on Ethereum), enabling the deployment of complex contracts and market designs through computer programs to automate transactions.

A combinatorial prediction market allows participants to bet on \emph{combinations of outcomes}.
For example, the opening value of AMZN will fall between 190 and 200 \emph{and} the opening value of GOOGL will be between 180 and 190 on July 18, 2025; the same political party will win both Ohio and Pennsylvania.  
Tailored to the designer's need, these markets may provide more information and refined forecasts by aggregating traders' predictions on a collection of events of interest, known as a set system.  
However, designing such markets is complex, particularly due to the algorithmic challenge of pricing and trading of the exponentially large number of possible combinations.

We are interested in designing \emph{automated market makers} (AMMs) and characterizing their computational complexity for combinatorial markets. 
An AMM uses predefined functions (e.g., cost functions~\cite{AbernethyChVa11,DBLP:journals/corr/abs-1206-5252}) to set prices and execute trades for any requested security; it removes the need for a traditional counterparty and can address the \emph{thin market} problem that is particularly pronounced in combinatorial markets.
The design of AMMs is critical, as it impacts how effectively markets integrate information and operate computationally.

Prior research has focused on and made substantial progress in understanding combinatorial prediction markets within specific set systems (e.g., intervals~\cite{Dudik21}, permutations~\cite{chen2007betting}), tackling design challenges by leveraging the unique structures of their securities or outcomes.  
However, a unified framework for markets associated with \emph{arbitrary} set systems remains elusive and yet to be developed. 
For a market designer, knowing the complexity of different combinatorial structures and their expressiveness is important for assessing the feasibility of implementing prediction markets in practice.
Our paper aims to fill this gap by exploring the algorithmic problem of computing security prices for an arbitrary set system from a geometric viewpoint.   
Specifically, we focus on the design of (efficient) AMMs for combinatorial prediction markets that support price queries, cost calculations, and always offer to buy or sell any combinatorial security at some price.
\emph{By establishing a novel connection between the market design problem and computational geometry, we present a unified framework for both analyzing the computational complexity and designing efficient algorithms for combinatorial prediction markets.}

Finally, we extend our framework beyond prediction markets to design AMMs for decentralized finance, which facilitate the trading of assets or cryptocurrencies (e.g., ERC-20 and BEP-20 tokens) without relying on a trusted third party. 
The core component of these exchanges is a class of AMMs known as \emph{constant function market makers} (CFMMs)~\cite{egorov2019stableswap,zhang2018formal,adams2021uniswap,martinelli2019balancer,zhang2018formal}, which support swap operations by determining the required amount of one asset in exchange for a specific amount of another. 
However, when it comes to multi-asset trades~\cite{angeris2021analysis}, where traders offer a basket of multiple assets to the CFMM in exchange for another basket of assets, the market making process becomes significantly more complex due to the vast number of possible combinations of asset baskets (450,000 assets for ERC-20 and over one million for BEP-20).  
We formalize the swap problem for multi-asset trades, demonstrate that it can be reduced from range update problems, and thus provide a gateway for understanding and designing algorithms for these trades.

%
\subsection{Our contribution}
We report a systematic study of the design of AMMs for various markets. 
We start with the popular \emph{logarithmic market scoring rule} AMMs on prediction markets, extend to other scoring rules, and finally, explore AMMs in decentralized finance. 
Below, we provide a brief summary of our results. 

\paragraph{Prediction markets using logarithmic market scoring rule.}
We first examine AMMs for prediction markets with Hanson's \emph{logarithmic market scoring rule} (LMSR).
LMSR has been extensively studied~\cite{Hanson03,Hanson07,chen2007betting,Dudik21} and widely deployed in practical contexts, such as predicting political events~\cite{hanson1999}, building opening date~\cite{Othman10}, product sales~\cite{Chen2002}, and instructor ratings~\cite{Chakraborty13}. 
In \cref{sec:equiv}, we demonstrate that the problem of AMMs for LMSR prediction market, which supports price, cost, and buy operations, is equivalent to the \emph{range query with multiplication range update} problem defined in \cref{sec:rqru}.  
This equivalence enables us to use tools from computational geometry to derive both algorithmic and hardness results for LMSR prediction markets across different set systems, as summarized in \cref{tab:results}.
In \cref{sec:alg}, we introduce a partition-tree-based scheme to design efficient LMSR algorithms for combinatorial prediction markets associated with set systems of finite VC dimensions.  
\cref{sec:hardness} provides our hardness results for other LMSR markets.  
Below, we highlight three results for LMSR: 
\begin{enumerate}
    \item We develop a partition-tree-based algorithm for LMSR markets on interval securities, where the outcome is a real-valued random variable, and traders can bet on interval events (\cref{ex:one}). 
    Our algorithms can support all market operations in time logarithmic in the size of outcome space $n$ (\cref{prop:one}), consistent with prior research~\cite{Dudik21}.  
    We further provide a lower bound in \cref{prop:hard1d} demonstrating the optimality of logarithmic time.
    \item Our partition-tree algorithm extends naturally to LMSR for $d$-dimensional outcome spaces (\cref{ex:orthogonal}), achieving sublinear running time of $O(n^{1-1/d})$ (\cref{prop:orthogonal}). 
    We show that achieving a sub-polynomial time for LMSR markets in even the two-dimensional setting is improbable (\cref{prop:hard2d}). 
    Otherwise, solving matrix multiplication in near-quadratic time would be feasible, contradicting the current leading algorithm, which requires time in $O(m^{2.371552})$~\cite{williams2023new}. 
    Our result provides an answer to the open problem proposed in \citet{Dudik21}.
    \item For general set systems, we show that combinatorial prediction markets that admit efficient algorithms can be characterized by the VC dimension of the set system (\cref{def:vc}). 
    Specifically, our partition-tree scheme has sublinear running time when the VC dimension is bounded. 
    Additionally, we provide an information-theoretic lower bound showing the non-existence of sublinear time algorithms when the VC dimension is unbounded (\cref{prop:limit_vc}).
    With this hardness result, we revisit the \#P-hardness results from \citet{chen2007betting} for boolean function securities on $\{0,1\}^K$ and pairing securities on permutations of $K$ candidates. 
    We prove that there is no $o(2^K)$ time LMSR for the boolean function securities and no $o(K!)$ time LMSR for the pairing securities (\cref{cor:vc}).  
\end{enumerate}

\begin{table}[]
\caption{Summary of AMMs for prediction markets with LMSR on $(\mathcal{X}, \mathcal{F})$ where $n = |\mathcal{X}|$ and $d\ge 2$.}
    \label{tab:results}
    \centering
    \begin{tabular}{p{0.35\textwidth} l p{0.35 \textwidth}}
    \toprule
        Set systems $(\mathcal{X}, \mathcal{F})$ & Running time  & Results\\
        \toprule
        Intervals (\cref{ex:one}) & $\Theta(\log n)$   & \cref{fig:tree,prop:one,prop:hard1d}  \\
        $d$-orthogonal sets (\cref{ex:orthogonal}) & $O(n^{1-1/d})$   & \cref{prop:orthogonal,prop:hard2d}\\
        Hyperplane in $\R^d$  (\cref{ex:hyperplane}) &$O(n^{1-1/d})$  & \cref{prop:hyperplance}\\
        Finite VC   (\cref{ex:topK})  & $O(n^{1-\epsilon})$ with $\epsilon>0$ & \cref{prop:abstract} \\
        Infinite VC (\cref{ex:pairing,ex:boolean})& no $o(n)$ & \cref{prop:limit_vc,cor:vc}\\
        \bottomrule
    \end{tabular}
\end{table}
\paragraph{Other market scoring rules.}
We extend our approach to other scoring rules for prediction markets, illustrating their connections to range query range update problem (RQRU) with varying update rules and employing variations of the partition tree scheme (\cref{sec:beyondlmsr}), as summarized in \cref{tab:results2}.  
We first note that the equivalence between a market maker and a certain RQRU problem is not sufficient for an efficient algorithm, as a single query problem can be NP-hard. 
However, we show that several common market makers can be reduced to RQRU problems that admit our partition tree scheme, thus enabling efficient algorithms.

\cref{sec:qmsr} studies the \emph{quadratic market scoring rule} (QMSR), another widely-adopted proper scoring rule~\cite{DBLP:journals/corr/abs-1206-5252}. 
Interestingly, differed from the LMSR reduction, we show that QMSR can be reduced to a range query with \emph{addition range updates} problem in \cref{lem:equivq}, which enables us to apply our partition tree scheme to solve QMSR with the same computational complexity (\cref{thm:qmsr}).   
In contrast to our hardness result regarding the subpolynomial time algorithm for LMSR on two-dimensional interval securities, we introduce a polylogarithmic time algorithm for QMSR in \cref{cor:orthogonalp}. 
\emph{This result underscores a computational distinction between LMSR and QMSR.}  

We extend to examine the \emph{power market scoring rules}~\cite{dawid2007geometry,jose2008scoring}, which include QMSR as a special case.  
In \cref{sec:pmsr}, we consider a power scoring rule with a degree of $\frac32$ and formulate the market making problem as a range query range update problem with \emph{group action updates} in \cref{lem:equivp}.  
We show that our partition-tree algorithm remains applicable under group action updates, preserving the same computational complexity (\cref{thm:pmsr}).

\cref{sec:multi-res} examines the \emph{multi-resolution market design} that can be naturally integrated into our partition-tree scheme, facilitating multiple market makers to mediate submarkets of increasingly fine-grained outcome partitions. 
We demonstrate that with efficient and local weight updates, such multi-resolution design will not affect our characterization of market complexity, including removing arbitrage that may arise due to the use of different market scoring rules for combinatorial securities associated with information at different granularity.\looseness=-1

\paragraph{Constant function market maker in decentralized finance.}
Finally, in \cref{sec:swap}, we discuss AMMs in decentralized finance, specifically the \emph{constant function market makers} (CFMM). 
We demonstrate that the swap operations problem can be reduced from range update problems, a special case of RQRU. 
Consequently, we show that under linear and logarithmic trading functions, a similar partition tree scheme can be applied to achieve the same computational complexity.



\begin{table}[]
\caption{Summary of reductions for AMMs to various range query range update problems}
    \label{tab:results2}
    \centering
    \begin{tabular}{l l l}
    \toprule
        Automated market maker & Query  & Update\\
        \midrule
        LMSR (\cref{def:lmsr}) & addition $+$   & multiplication $\cdot$\\
        QMSR (\cref{def:qmsr}) & addition $+$   & addition $+$\\
        $\frac32$-MSR (\cref{def:pmsr}) & addition $+$  & group action\\
        \midrule
        Log CFMM (\cref{eq:swap_log})  & addition $+$ & multiplication $\cdot$ \\
        Linear CFMM (\cref{eq:swap_lin}) & addition $+$ & addition $+$\\
        Geometric mean CFMM (\cref{eq:swap_prod}) & multiplication $\cdot$ & addition $+$\\
        \bottomrule
    \end{tabular}
\end{table}
\subsection{Related work}
\paragraph{Designing combinatorial prediction markets.}
While \citet{AbernethyChVa11} provide a thorough characterization of cost-based markets for combinatorial prediction markets that satisfies key axioms for eliciting truthful predictions, efficient algorithms for combinatorial prediction are not fully understood.\fang{check} 
\citet{Chen_CombBet(2007)} demonstrate that simple comparison securities on permutations (\cref{ex:pairing}) or Boolean function securities on hypercubes (\cref{ex:boolean}) are \#P-hard.  
\citet{Dudik21} present an efficient algorithm for interval securities (\cref{ex:one}).  \citet{DBLP:conf/stoc/ChenGP08} offer efficient algorithms for tournament outcomes where prices can be succinctly encoded as a Bayesian tree.  
Our established connection to computational geometry provides an algorithmic approach to extract structure of those securities.   

There are also prior works focusing on relaxation techniques. 
\citet{DBLP:conf/ijcai/XiaP11} provide a Monte Carlo algorithm for approximate pricing on tournament outcomes with Bayesian network distributions. 
\citet{DBLP:journals/jair/Laskey0HTMG18} generalize the result to Bayesian network structure preserving distributions. 
\citet{DudikLaPe12,dudik2013combinatorial} relax the arbitrage-free condition. 
\citet{kroer2016arbitrage} propose an integer-programming-based arbitrage removal algorithm but relax the worst-case computational complexity guarantee.

Recent works examining decentralized finance has introduced an axiomatic framework that connects general \textit{constant function market makers} (CFMMs), which form the core implementation of \textit{Uniswap v2}~\cite{adamsuniswapv2}, to cost-function-based prediction markets~\cite{Frongillo2023}.
The work opens up the possibility of designing and characterizing the complexity of CFMMs using results in combinatorial prediction markets.

\paragraph{Range query range update.}
In computational geometry, the range query problem has been extensively studied~\cite{agarwal2017range}, particularly in settings without updates or with point updates, where an update operation modifies the weight of a single point. 
Our range query range update problem requires support for updates on \emph{a set of points} (\cref{def:rqru}), generalizing the classical range query problem with point updates when the set system contains all singletons.
Several recent works have also explored range queries with range updates, with a primary focus on regular orthogonal set systems. 
For instance, \citet{DBLP:journals/corr/abs-2101-02003,mishra2013new} investigate addition range updates, while \citet{yang2020updating,DBLP:journals/corr/abs-2101-02003} delve into the hardness of general range updates. 
Techniques like lazy propagation have been effectively applied to one-dimensional or multi-dimensional orthogonal set systems \cite{yang2020updating,DBLP:journals/corr/abs-2101-02003}. 
Our partition-tree-based algorithm presents a novel adaptation for handling general set systems.




 









\section{Preliminaries}
An \emph{automated market maker} (AMM) is an algorithm that trades securities.  At a high level, a design problem for AMM needs to specify: 1) what securities can be traded, 2) how these securities are traded, and 3) which operations are supported.  The concept of AMMs has been implemented in various applications, including blockchain and prediction markets.   For prediction markets, an AMM trades prediction bets using a cost function and supports price, cost, and buy operations.  For blockchain, an AMM trades bundles of cryptocurrencies using a trading function, and supports swap operations.  We will introduce AMMs for prediction markets here and then AMMs for blockchain in \cref{sec:swap}.

\subsection{AMMs for prediction markets}
\label{sec:pre_lmsr}
A prediction market provides securities (prediction bets) to sequentially aggregates trader's prediction on a random variable.  We define the design problem of AMMs for prediction markets in \cref{def:generalcost,def:lmsr}.  Before introducing the problem, we will first introduce combinatorial securities, cost functions, and then price, cost, and buy operations.
\paragraph{Combinatorial securities on prediction markets}
Let $\mathcal{X}$ denote an outcome space with $n$ mutually exclusive possible \emph{outcomes}.  When $n$ is large, it is natural (both computationally and economically) to elicit probabilities on a set of events, denoted as $\mathcal{F}\subseteq 2^\mathcal{X}$ that is a collection of subsets of $\mathcal{X}$ called a \emph{set system}.
Consider a random variable on $\mathcal{X}$, say the value of S\&P500 at 4pm tomorrow, and $\mathcal{F}$ can be the collection of all intervals on $\mathcal{X}$, i.e., $E_{(i,j)} = [i,j]\cap \mathcal{X} = \{x\in \mathcal{X}: i\le x\le j\}$ where $i\le j\in \R$.
\cref{sec:example_sec} present more examples.

Given $\mathcal{F}$, a combinatorial prediction markets provides \emph{combinatorial securities} $\phi_E: \mathcal{X}\to \{0,1\}$ for all $E\in \mathcal{F}$, which is simply the \emph{payoff function} paying out \$1 if event $E$ occurred and \$0 otherwise.  Such collection of securities is known as \emph{combinatorial securities} for a set system $(\mathcal{X},\mathcal{F})$, or $(\mathcal{X}, \mathcal{F})$ securities for short.
A trader trades $s\in \R$ share of security $\phi_E$ with the central AMM (where positive $s$ corresponds to purchases and negative to short sales), and receives a payoff of $s \cdot\phi_E$.

The AMM adjusts the price of each security after trading with each trader so that the prices reflect the consensus predictions among traders. The price of each security can be viewed as the traders’ collective estimation of the probabilities of their associated events. To facilitate such a combinatorial market, an ideal AMM needs to both incentivize trades to incorporate new information and \emph{efficiently compute prices and market states} after each trade.

\paragraph{Cost-Function-Based prediction market}
A long line of work~\cite{AbernethyChVa11,chen2012utility} demonstrate how to trade these securities that achieve desirable incentive properties e.g., \emph{no arbitrage} and \emph{bounded loss}.
The \emph{no-arbitrage} property requires that as long as all outcomes $x$ are possible, there be no market transaction with a guaranteed profit for a trader. 
The \emph{bounded-loss} property is defined in terms of the worst-case loss of a market maker, i.e., the largest difference, across all possible trading sequences and outcomes, between the amount that the market maker has to pay the traders (once the outcome is realized) and the amount that the market maker has collected (when securities were traded).
The property requires that this worst-case loss be a priori bounded by a constant.

Following~\cite{AbernethyChVa11} and~\cite{chen2012utility}, we consider that the AMM determines security prices using a potential function $C \colon \R^{\card{\mathcal{X}}}\to\R$, called a \emph{cost function}.  
The \emph{market state} is specified by a vector $\vw\in \R^{\card{\mathcal{X}}}$, where we will use $w_x$ or $w(x)$ to denote the number of shares of security sold by the AMM for the outcome $x\in \mathcal{X}$.\footnote{In general, a cost function is $\tilde{C}:\R^{\card{\mathcal{F}}}\to \R$ where a state $\tilde{\vw}\in \R^{\card{\mathcal{F}}}$ stores the number of share on each possible event $E\in \mathcal{F}$~\cite{AbernethyChVa11}.  
In \cref{sec:lmsr,sec:qmsr,sec:pmsr}, we consider cost function that has $C: \R^{\card{\mathcal{X}}}\to \R$ such that $\tilde{C}(\tilde{\vw}) = C\left(\sum_{E\in \mathcal{F}}\tilde{w}_x\mathbf{1}_E\right)$ for all $\tilde{\vw}\in \R^{\card{\mathcal{F}}}$, and \cref{sec:multi-res} studies the multi-resolution markets of the general form.}  
Below are some popular cost functions that we study in this paper, which all satisfy the properties of no-arbitrage and bounded loss.
\begin{itemize}
    \item The \emph{logarithmic market scoring rule} (LMSR) proposed by ~\citet{Hanson03,Hanson07} is a popular cost-function-based market making mechanism.
It uses the logarithmic scoring rule
to interact with a sequence of traders to trade securities and maintain a probability distribution over an outcome space $\mathcal{X}$ of $n$ points in an online manner. \begin{equation}
\label{eq:LMSR}
C_L(\vw)=b\ln \left(\sum_{x\in\mathcal{X}} e^{w_x / b}\right).
\end{equation}
where $b>0$ is the liquidity parameter.  
The market designer can choose $b$ to  control how fast the price moves in response to trading and limiting the worst-case loss of the market maker to $b\ln\,\card{\mathcal{X}} = b\ln (n)$~\cite{Hanson03}.
    \item The quadratic scoring rule (QMSR) is another popular choice of proper scoring rule with the following cost function \cite{DBLP:journals/corr/abs-1206-5252}\begin{equation}
        \label{eq:QMSR}
        C_Q(\vw)=\frac{1}{n}\sum_{x\in \mathcal{X}}w_x+\frac{1}{4b}\sum_{x\in \mathcal{X}}w_x^2-\frac{1}{4bn}\left(\sum_{x\in \mathcal{X}}w_x\right)^2-\frac{b}{n}
    \end{equation}
    where $b>0$ is the liquidity parameter.
    \item A $\gamma$-power MSR has the following cost function~\cite{dawid2007geometry,jose2008scoring}   \begin{equation}\label{eq:pmsr}
    C_\gamma(\vw) = \max_{p\in \Delta_\mathcal{X}} \sum_{x\in \mathcal{X}} w_xp(x)-b\sum_{x\in \mathcal{X}}p(x)^\gamma
\end{equation}
which includes the above QMSR as a specific case when $\gamma = 2$, as shown in~\cite{DBLP:journals/corr/abs-1206-5252}.
\end{itemize}
More generally, a cost function can be any function that is convex,  differentiable, and $\mathbf{1}$-invariant so that $C(\vw+s\mathbf{1}) = C(\vw)+s$ for all $\vw\in \R^{\card{\mathcal{X}}}$ and $s\in \R$.~\cite{AbernethyChVa11}

\paragraph{Operations on combinatorial securities}

An AMM needs to support operations for trading on securities in $\mathcal{F}$ defined below.
\begin{Definition}\label{def:generalcost}
    Given a set system $(\mathcal{X}, \mathcal{F})$ and a cost function $C: \R^{|\mathcal{X}|}\to \R$, an AMM for prediction market with cost function $C$ on $(\mathcal{X}, \mathcal{F})$ securities takes an initial state $\vw^{(0)}:\mathcal{X}\to \mathbb{R}$ and offers securities for all $E\in \mathcal{F}$, supports a sequence of $\price$, $\cost$, and $\buy$ operations taking one of the following forms: for any set $E\in \mathcal{F}$, shares $s\in \mathbb{R}$, and state $\vw$,
\begin{itemize}
\item Price operation returns the current price of security for $E$, $\price_C(E; \vw) = \sum_{x\in E}\frac{\partial}{\partial w_x}C(\vw).$
    \item Cost operation returns the current cost of $s$ shares of security for $E$, $
    \cost_C(E, s; \vw) = C(\vw+s\mathbf{1}_E)-C(\vw)$
where $\mathbf{1}_E$ is the indicator vector of set $E$. 
    \item Buy operation, $\buy_C(E, s;\vw)$, updates the state $\vw\gets \vw+s\mathbf{1}_E$. 
\end{itemize}
\end{Definition}
Thus, a trader who wants to buy one share combinatorial security $\phi_E$ in the market state $\vw$ must pay $C(\vw+\mathbf{1}_E)-C(\vw)$ to the market maker, after which the new state becomes $\vw+\mathbf{1}_E$.
The vector of instantaneous prices in the corresponding state $\vw$ is $\nabla C(\vw)$.

In this paper we will focus on AMMs for prediction markets with LMSR which are special cases of \cref{def:generalcost} by taking $C = C_L$ in \cref{eq:LMSR}.
\begin{Definition}[LMSR market]\label{def:lmsr}
    Given a set system $(\mathcal{X}, \mathcal{F})$, an AMM for {a prediction market with LMSR on $(\mathcal{X}, \mathcal{F})$} takes an initial state $\vw^{(0)}:\mathcal{X}\to \mathbb{R}$ and offers securities for all $E\in \mathcal{F}$, supports a sequence of $\price$, $\cost$, and $\buy$ operations taking one of the following forms: for any set $E\in \mathcal{F}$, shares $s\in \mathbb{R}$, and state $\vw$,
\begin{itemize}
\item $\price(E; \vw)$: return the current price of security for $E$, 
\begin{equation}
    \price(E; \vw) = \frac{\sum_{x\in E}e^{w_x/b}}{\sum_{x'\in \mathcal{X}} e^{w_{x'}/b}} = \sum_{x\in E}\frac{\partial}{\partial w_x}C_L(\vw).\label{eq:def_price}
\end{equation}
    \item $\cost(E, s; \vw)$: return the current cost of $s$ shares of security for $E$, \begin{align}
    \cost(E, s; \vw) = C_L(\vw+s\mathbf{1}_E)-C_L(\vw) &= b\ln\left(e^{s/b}\price(E;\vw)+ 1 - \price(E;\vw)\right).
    \label{eq:def_cost}
\end{align} 
    \item $\buy(E, s;\vw)$: update the state $\vw\gets \vw+s\mathbf{1}_E$. 
\end{itemize}
Moreover, when there is no ambiguity, we refer to a such AMM as an LMSR market, LMSR algorithm, or simply LMSR.
\end{Definition}
We use $\vw$ to denote a generic symbol for the number of security and add superscript $t$ to emphasize the state at round $t$, $\vw^{(t)}$.  
We may omit $\vw$ and write $\price(E), \cost(E,s)$, and $\buy(E, s)$ when there is no ambiguity. 

Note that \cref{def:lmsr} is an online algorithm problem.  
Specifically, given a collection of securities $\mathcal{F}$ on $\mathcal{X}$, we aim to prepare a data structure to store auxiliary information that can facilitate responding to a sequence of operations (price, buy, and cost) efficiently.  
We measure the computational complexity by the time spent for each operation in the worst case.  Formally, we say an LMSR market that can support price operation in $T_P(n)$, cost operation in $T_C(n)$, and buy operation in $T_B(n)$, if the time spent on each operation is always upper bounded by those values for any sequence of operations.  
Additionally, the \emph{running time} of an LMSR is $\max\{T_P(n), T_C(n), T_B(n)\}$.  
We note that for LMSR all operations can be trivially done in linear time by exhausting all $n$ outcomes.  However, as the number of outcomes $n$ becomes large, it becomes critical to achieve \emph{sublinear} or even polylogarithmic running time.  Finally, a linear-size data structure and preparation time for \cref{def:lmsr} is inevitable in the worst-case scenario.  However, when the initial condition is uniform, we may speed up the preparation time.




\subsection{Examples of combinatorial securities}
\label{sec:example_sec}
We first introduce several examples of combinatorial securities for set systems, and define metric (VC dimension and dual shattering dimension) to measure the complexity of set systems.
\begin{example}[Interval security~\cite{Dudik21}]\label{ex:one} 
The outcome space is $\mathcal{X} \subset \R$ with $|\mathcal{X}| = n$, and $\mathcal{F}$ is the collection of all intervals $E_{(i,j)} = [i,j]\cap \mathcal{X} = \{x\in \mathcal{X}: i\le x\le j\}$ where $i\le j\in \R$.  
Each interval security corresponds to predictions that the outcome will fall into the specified interval. 
Though $\mathcal{X}$ can be any collection of $n$ real-valued points, and without loss of generality, we can scale any real line, so that $\mathcal{X} = [n] = \{0,1\dots, n-1\}$.

For example, we may construct a prediction market for the S\&P 500 opening price on Oct 18, 2025, by setting $n = 2^{20} = 1048576$ and $\mathcal{X} = \{0, 0.01, \dots, 10485.74, 10485.75\}$, where we cap prices at $10485.75$.  
Then an interval security corresponds to the opening price falling into a specific interval.


\end{example}

\begin{example}[Multi-dimensional orthogonal security]\label{ex:orthogonal}
Given a positive integer $d$, the outcome space is $\mathcal{X}\subset \R^d$ with $|\mathcal{X}| = n$. 
Each security is represented as an axis-aligned hyperrectangle, $E_{\vi, \vj} = [i_1, j_1]\times\dots\times [i_d, j_d]\cap \mathcal{X}$, where $\vi = (i_1, \dots, i_d), \vj = (j_1, \dots, j_d)$ and $\vi$ is less than $\vj$ at all coordinates denoted as $\vi\le \vj$, and the $d$-dimensional orthogonal set system is $\{E_{\vi, \vj}: \vi\le \vj\in \R^d\}$.  If $\mathcal{X} = [m]^d$ for some $m$, we call $(\mathcal{X}, \mathcal{F})$ a \emph{regular} orthogonal securities.  A $d$-dimensional orthogonal security is a natural generalization of interval security, and thus we also name them $d$-dimensional interval securities.  Instead of S\&P 500, we may want to predict the opening prices of the top five companies by market cap (MSFT, AAPL, NVDA, AMZN, and GOOGL).  Multi-dimensional orthogonal securities with $d = 5$ allow traders to express predictions of events where each price falls within specific intervals.
\end{example}

\begin{example}[Hyperplane security~\cite{Guo_CombPred(2009),combo_options,pred_options}]\label{ex:hyperplane} 
Similar to \cref{ex:orthogonal}, the outcome space is $\mathcal{X}\subset \R^d$ with $|\mathcal{X}| = n$, and a hyperplane security is associated with a hyperplane with $\beta\in \R^d, \beta_0\in \R$ where $E_{\beta, \beta_0} = \{x\in \mathcal{X}: \beta^\top x+\beta_0\ge 0\}$.  A hyperplane security represents a linear combination of multi-dimensional outcomes.  In particular, many interval securities on index funds, including S\&P 500, can be represented as a hyperplane security (e.g., the opening price of the S\&P 500 is less than some constant $\beta_0$).
In financial options, investors speculate on whether the underlying security (bundle) will be realized with a greater or smaller value than \textit{the strike price} on the expiration date.


\end{example}

Although the above three examples have Euclidean outcome space, the combinatorial securities can be designed for abstract set system.  We introduce two additional scenarios where the outcome space can be permutations or hypercubes~\cite{Chen_CombBet(2007),DBLP:journals/corr/abs-0802-1362}.

\begin{example}[Top $L$ candidates]\label{ex:topK}
Top $L$ securities allows traders to bet on the outcome of top $L$ candidates among $K\ge L$ candidates.  Given a positive integer $L\le K$ and a set $H\subset [K]$ with $|H| = L$, a top $L$ security for $H$ is the set of permutations where the top $L$ candidates are from $H$. There are a total of $\binom{K}{L}$ top $L$ securities and $n := K!$ possible outcomes. 
\end{example}

\begin{example}[Permutations~\cite{Chen_CombBet(2007)}]\label{ex:pairing}
Pairing securities allows traders to bet on whether one candidate will rank higher than another candidate, where the outcomes are permutations of $K$ candidates. For all $i\neq j\in [K]$, a pairing set
$\tau_{(i,j)}$ is the set of permutations where $i$ ranked higher than $j$.  There are a total of $K(K-1)$ different pairing sets in $\mathcal{F}$ and $n := K!$ possible outcomes. 
\end{example}

\begin{example}[Boolean function\cite{DBLP:journals/corr/abs-0802-1362}]\label{ex:boolean}
    Given $K$, the outcome space is $\mathcal{X} = \{0,1\}^K$, and each securities corresponds a boolean function $\psi:\mathcal{X}\to \{0,1\}$.  A function is $L$-junta if the output only depends on $L$ coordinates of the input.  For instance, $1$-junta functions are $\psi_i(x) = x_i$, for all $x = (x_1, \dots, x_K)$ and $i = 1,\dots, K$.  We further define the $1$-junta set system as $\{E_i: i = 1,\dots, K\}$ so that $E_i = \{x\in \mathcal{X}: \psi_i(x) = 1\}$, and call the corresponding securities as $1$-junta securities.  
 Note that the disjunctions of two coordinate securities in \cite{DBLP:journals/corr/abs-0802-1362} belongs to $2$-junta securities and contains our $1$-junta securities as special case. 
\end{example}

\subsection{Complexity of set systems}
In addition to the above examples, we will leverage the VC dimension and dual shattering dimension to measure the complexity of general set systems and demonstrate how to use these measures to characterize the computational complexity of the AMM design problem.
\begin{Definition}
    \label{def:vc}
    \xw{adapted from the original VC-d paper?}\fang{I am not sure.  I can find a reference for this}
    Given a set system $(\mathcal{X}, \mathcal{F})$ and $\mathcal{X}'\subseteq \mathcal{X}$, let $\Pi_\mathcal{F}(\mathcal{X}') = \{\mathcal{X}'\cap E: E\in \mathcal{F}\}$ be all intersections between $\mathcal{X}'$ and elements of $\mathcal{F}$. 
    We say $\mathcal{X}'$ is \emph{shattered} by $\mathcal{F}$ if $\Pi_\mathcal{F}(\mathcal{X}') = 2^{\mathcal{X}'}$, the power set of $\mathcal{X}'$.  
    The \defn{Vapnik-Chervonenkis dimension} of $(\mathcal{X}, \mathcal{F})$, or VC-dimension 
for short, is the size of the largest set $\mathcal{X}'$ that is shattered by $\mathcal{F}$.
\end{Definition}
We say the collection of  securities has a finite (or bounded) VC dimension if the associated set system  has a finite VC dimension.  
It is well-known that the VC dimensions of \cref{ex:one,ex:orthogonal,ex:hyperplane,ex:topK} are all finite when $d$ and $L$ are finite.  We will further show that the VC dimensions of \cref{ex:pairing,ex:boolean} are infinite in \cref{cor:vc}.

Similar to the VC dimension, the dual shattering dimension measures the complexity of a set system. 
\begin{Definition}\label{def:shattering}
    Given a set system $(\mathcal{X}, \mathcal{F})$, let $\mathcal{A}\subseteq \mathcal{F}$ be a subset of $\mathcal{F}$.  Two points $x, y\in \mathcal{X}$ are $\mathcal{A}$-equivalent if $x\in E\Leftrightarrow y\in E$ for any $E\in \mathcal{A}$.  For an integer $m$, the \emph{dual shatter function} $\pi^*_\mathcal{F}(m)$ is the maximum number of $\mathcal{A}$-equivalence classes on all possible $m$-set $\mathcal{A}\subseteq \mathcal{F}$, and the \defn{dual shattering dimension} is the smallest $d$ such that $\pi^*_\mathcal{F}(m) = O(m^d)$.  
\end{Definition}

Noteworthy, a set system has a finite VC dimension if and only if the dual shattering dimension is finite.\footnote{Specifically, if a set system has VC dimension equals $D$, the dual shattering dimension $d$ is bounded by $2^{D+1}$.  
Conversely, if the set system has $d$ dual shattering dimension, the VC dimension $D$ is bounded by $O(d^{O(d)})$~\cite{assouad1983densite,Chazelle1989}.  
Moreover, the dual shattering dimension might be smaller than the VC dimension of the range space.  
Indeed, in the case of spheres and hyperplanes in $\R^d$, the dual shattering dimensions are just $d$, while the VC dimensions are both $d+1$.}  
In geometric settings, bounding the dual shattering dimension is relatively easy, as it depends on the complexity of the arrangement of $m$ ranges of this space.  
For instance, the dual shattering dimension of lines on $\R^2$ is $2$, because the maximum number of distinct regions partitioned by $m$ lines is $\frac{m^2+m+2}{2} = O(m^2)$.  
The dual shattering dimension for hyperplane on $\R^d$ is $d$, because the distinct regions partitioned by $m$ hyperplanes is $\sum_{i = 0}^d\binom{m}{i}$~\cite[Proposition 2.4.]{stanley2004introduction}.  Similarly, the dual shattering dimension of $d$-spheres on $\R^d$ is $d$, because the distinct regions partitioned by $m$ sphere is $\binom{m-1}{d}+\sum_{i = 0}^d\binom{m}{i}$~\cite{2832639}.  






\subsection{Range query and range update (RQRU)}\label{sec:rqru}
Range query is a classical problem in computational geometry with many variants~\cite{DBLP:books/lib/Berg97}.  
A typical range-query problem on a set system $(\mathcal{X}, \mathcal{F})$ tries to address questions regarding elements $E$ of $\mathcal{F}$.  For example, one might want to count the number of points within each $E$ or compute the sum of weights associated with all points in $E$.
Here we introduce one variant for LMSR, and we will extend the definition for other market scoring rule settings in \cref{sec:beyondlmsr}, and finally general RQRU in \cref{app:rqrug}.  

Given a set system $(\mathcal{X}, \mathcal{F})$, each point $x\in \mathcal{X}$ is assigned a positive weight $W(x)\in \mathbb{R}_+$, where $\R_+$ is the set of positive real numbers.  For any subset (range) $E\subseteq \mathcal{X}$, let $W(E):=\sum_{x\in E} W(x)$.  Range query problems ask for algorithms that preprocesses a set system $(\mathcal{X}, \mathcal{F})$ into a data structure that computes and updates the weight $W(E)$ efficiently for any range $E\in \mathcal{F}$.  Formally, 
\begin{Definition}\label{def:rqru}
    The \defn{range query with multiplication range update problem}, $(+,\cdot)$-RQRU for short, gives a set system $(\mathcal{X}, \mathcal{F})$ and initial weights $W^{(0)}:\mathcal{X}\to \mathbb{R}_+$.  It requests a sequence of operations, taking one of the following forms: for any $E\in \mathcal{F}$ and $S\in \mathbb{R}_+$:
\begin{itemize}
    \item $\query(E; W)$: compute and return the total weight of range $E$, $W(E) = \sum_{x \in E} W(x)$.
    \item $\update(E, S; W)$: for each $x \in E$, update $W(x) = S\cdot W(x)$, and for each $x'\notin E$, $W(x') = W(x')$.  
\end{itemize}
\end{Definition}
Similarly, we will use $W$ to denote a generic symbol for the weights of each points and add superscript $t$ to emphasize the state at round $t$, $W^{(t)}$.  We may omit $W$ and write $\query(E), \update(E,S)$ when there is no ambiguity.  
We will use RQRU to refer $(+,\cdot)$-RQRU.  In later section, we will generalize it to $(\oplus, \otimes)$-RQRU where the query uses function operator $\oplus$ and update function uses $\otimes$ defined in \cref{app:rqrug}

We measure the performance of a data structure by the time spent for each operation when the size of $\mathcal{X}$ is $n$.  Specifically, let $T_Q(n)$ be the \emph{query time} to support the range query, $T_U(n)$ be \emph{update time} for the updates, $\max\{T_Q(n), T_U(n)\}$ be the \emph{running time}.  Finally, the time for initialized the data structure is called \emph{preprocessing time} $T_I(n)$ which is generally less critical since the data structure is constructed only once.

\section{Algorithmic and hardness results for LMSR}\label{sec:lmsr}

We first establish equivalence between LMSR and RQRU in \cref{thm:equiv}.  
Then, we delve into the exploration of possibilities and limitations associated with LMSR.  \Cref{sec:alg} introduces a general framework from computational geometry, partition tree, for designing efficient LMSR algorithms.
Using \cref{thm:equiv}, we provide some hardness results for LMSR algorithms in \cref{sec:hardness}.

\subsection{Equivalence between LMSR and RQRU}\label{sec:equiv}
One main contribution is establishing an equivalence between LMSR in \cref{def:lmsr} and RQRU in \cref{def:rqru} which enables us to leverage tools from computational geometry to derive algorithmic as well as hardness results for LMSR.

\begin{theorem}\label{thm:equiv}
    For any set system $(\mathcal{X}, \mathcal{F})$ with $|\mathcal{X}| = n$, if there is a $(+,\cdot)$-RQRU algorithm on $(\mathcal{X}, \mathcal{F})$ with $T_Q(n)$ query time and $T_U(n)$ range update time, there exists a LMSR algorithm on $(\mathcal{X}, \mathcal{F})$ that can support price operation in $2T_Q(n)+1$, buy operation in $2T_Q(n)+T_U(n)+2$, and cost operation in $2T_Q(n)+2T_U(n)+7$ using the same order of space. 
    
    Conversely, if there is an LMSR algorithm on $(\mathcal{X}, \mathcal{F})$ that can support price operation in $T_P(n)$ and buy operation in $T_B(n)$, there is a $(+,\cdot)$-RQRU algorithm  on $(\mathcal{X}, \mathcal{F})$ with $T_P(n)+1$ query time and $T_P(n)+T_B(n)+4$ range update time using the same order of space.
\end{theorem}
The above reduction applies for all possible set system $(\mathcal{X}, \mathcal{F})$ and has asymptotic tight overhead where the running time of an LMSR can be of the same order as the running time of an RQRU algorithms. 
The key observation is that only the buy operation can alter the state of the market, and we only need to maintain sufficient information to address price and cost operations after each buy operation.  For LMSR, maintaining $\sum_{x\in E} e^{w(x)/b}$ for all $E\in \mathcal{F}$ is sufficient.  Later in \cref{sec:beyondlmsr}, we will show how to extend this idea to prediction markets for other scoring rules.  

\begin{proof}[Proof of Theorem~\ref{thm:equiv}]

  We first show a reduction from a LMSR market to a RQRU algorithm.  Given an initial state (the numbers of outstanding securities) $\vw^{(0)}$ on $(\mathcal{X}, \mathcal{F})$, we run RQRU on initial weight $W^{(0)}$ where $W^{(0)}(x) = e^{bw^{(0)}_x}$ for all $x\in \mathcal{X}$ and store an additional variable $M:=\sum_x W^{(0)}(x)$.  
    \begin{itemize}
        \item For each price operation with $E\in \mathcal{F}$, we return $\query(E)/M$ by calling the range query function from the RQRU algorithm.  
        \item For each buy operation with $E\in \mathcal{F}$ and share $s\in \mathbb{R}$, we run the following four steps: compute $a = \query(E)$, run update with set $E$ and $e^s$, $\update(E, e^{bs})$, $a' = \query(E)$, and $M\gets M-a+a'$.   
        \item Finally, to compute a cost operation with set $E$ and share $s$, we run the following three steps: First, run $\update(E,e^{bs})$ and compute $c' = \ln(\query(\mathcal{X}))$.  Second, run $\update(E, e^{-bs})$ and compute $c = \ln(\query(\mathcal{X}))$. Third, return $c'-c$. 
    \end{itemize}
    Note that a price operation takes one range queries with one arithmetic operation (division), a buy operation takes one update query, two queries and two arithmetic operation (one exponentiation and one multiplication), and a cost operation takes two range queries, two update queries, and seven arithmetic operations (one subtraction, two multiplications, two log, and two exponentiation), which proves the time complexity.  

    To prove the correctness, we first use induction on the sequence of operations to show that the weights in RQRU always equals exponential of the shares in the LMSR market for all round $t$,
    \begin{equation}
        M^{(t)} = \sum_{x\in \mathcal{X}}e^{bw^{(t)}_x} \text{ and } W^{(t)}(x) = e^{bw^{(t)}_x}\text{ for all }x\in \mathcal{X}\label{eq:red_inv}
    \end{equation}  
    where subscript $t$ emphasizes the variable at round $t$. 
    
    The based case holds by initialization.  If we encounter a buy operation $\buy(E, s)$ at round $t+1$, the share of $x\in E$ is updated from $w^{(t)}_x$ to $w^{(t+1)}_x = w^{(t)}_x+s$, and the above reduction also updates $W^{(t)}(x)$ to $W^{(t+1)}(x) = W^{(t)}(x) e^{bs} = e^{bw^{(t)}_x+bs} = e^{bw^{(t+1)}_x}$.  The equality also holds for all $x\notin E$.  Moreover, because $a = \sum_{x\in E} e^{bw^{(t)}_x}$ and $a' = \sum_{x\in E} e^{bw^{(t+1)}_x}$, $M = \sum_{x\in \mathcal{X}} e^{bw^{(t)}_x}-\sum_{x\in E} e^{bw^{(t)}_x}+\sum_{x\in E} e^{bw^{(t+1)}_x} = \sum_{x\in \mathcal{X}} e^{bw^{(t+1)}_x}$.  Therefore, we prove \cref{eq:red_inv} as other two operations do not change the state $W^{(t+1)} = W^{(t)}$ and $w^{(t+1)} = w^{(t)}$.  We then show the reduction answers price and cost queries correctly.  Given a price operation with $E$ at round $t$, the reduction returns 
    $$\frac{\query(E)}{M} = \frac{\sum_{x\in E}W^{(t)}(x)}{\sum_{x\in \mathcal{X}}W^{(t)}(x)} = \frac{\sum_{x\in E}e^{bw^{(t)}_x}}{\sum_{x\in \mathcal{X}}e^{bw^{(t)}_x}}$$ 
    which equals $\price(E;\vw^{(t)})$ in \cref{eq:def_price}.  Given a cost operation with $E$ and $s$ share at round $t$, the reduction computes $c' = \ln\left(\sum_{x\in E}W^{(t-1)}(x)e^s+\sum_{x\notin E}W^{(t-1)}(x)\right)$ and $c = \ln\left(\sum_{x\in E}W^{(t-1)}(x)+\sum_{x\notin E}W^{(t-1)}(x)\right)$.  Because $W^{(t-1)} = \vw^{(t)}$, 
    \begin{align*}
        c'-c = \ln\left(\sum_{x\in E}e^{bw^{(t)}_x}e^{bs}+\sum_{x\notin E}e^{bw^{(t)}_x}\right)-\ln\left(\sum_{x\in E}e^{bw^{(t)}_x}+\sum_{x\notin E}e^{bw^{(t)}_x}\right)
    \end{align*} which equals $\cost(E,s;\vw^{(t)})$ in \cref{eq:def_cost}. 

    For the other direction, if we have a LMSR market for $(\mathcal{X}, \mathcal{F})$, we construct RQRU with the following reduction:  Given an initial weight $W^{(0)}$, we create an additional normalizing variable $M$ with initial value equal to $\sum_x W^{(0)}(x)$ and run LMSR market with initial state $\vw^{(0)}$ and $b = 1$ so that $w^{(0)}_x = \ln W^{(0)}(x)$ for all $x\in \mathcal{X}$.
    \begin{itemize}
        \item For each range query with $E\in \mathcal{F}$, we return $M\cdot\price(E)$ by calling the price operation from the LMSR market algorithm.  
        \item For each update query with set $E\in \mathcal{F}$ and $S\in  \mathbb{R}_{>0}$, we first update the normalizing variable $M$ to $M\left(1+\price(E)(S-1)\right)$, and then run the buy operation with set $E$ and $\ln S$, $\buy(E, \ln S)$, from the LMSR market algorithm.
    \end{itemize} 
    Similar to the first part, the reduction has the following following invariant
    \begin{equation}
        M^{(t)} = \sum_x W^{(t)}(x)\text{ and }w^{(t)}_x = \ln W^{(t)}(x)\text{ for all }x\text{ and }t.\label{eq:red_inv2}
    \end{equation}
    To show the first part of \cref{eq:red_inv2}, we note that given an update query with $E\in \mathcal{F}$ and $S>0$ at round $t+1$, because $w^{(t+1)}_x = Sw^{(t)}_x$ if $x\in E$ and $w^{(t+1)}_x = w^{(t)}_x$ otherwise, $$M^{(t+1)} = \sum_x e^{w^{(t)}_x}\left(1+\frac{\sum_{x\in E}e^{w^{(t)}_x}}{\sum_x e^{w^{(t)}_x}}\right)(S-1) = \sum_{x\notin E}e^{w^{(t)}_x}+\sum_{x\in E}Se^{w^{(t)}_x} = \sum_x e^{w^{(t+1)}_x} = \sum_x W^{(t+1)}(x).$$
    The rest is similar to \cref{eq:red_inv}'s.  With \cref{eq:red_inv2,eq:def_price}, given a range query with $E$ at round $t$, the reduction returns 
    $M^{(t)}\price(E) = \sum_{x\in E} e^{w^{(t)}_x} = \sum_{x\in E} W^{(t)}(x)$ which completes the proof.
\end{proof}


\subsection{Partition tree scheme for LMSR}\label{sec:alg}
Now, we introduce the partition tree scheme, which has been extensively used in computational geometry, and design a lazy propagation algorithm on partition trees that supports $(+,\cdot)$-RQRU and thus LMSR. 
We introduce necessary notions for the partition tree scheme in \cref{sec:pre_partition} and define our lazy propagation algorithms for RQRU in \cref{sec:lazy}.  Then, we demonstrate the efficacy of our algorithms for LMSR summarized in \cref{tab:results}.  

\subsubsection{Partition tree scheme}\label{sec:pre_partition}
Partition tree scheme is a fundamental data structure for range query problem that contain one-dimensional segment tree (\cref{fig:tree}) and $k$-d trees as special cases.~\cite{willard1982polygon,Chazelle1989}  A partition tree utilizes the idea of recursively subdividing space into regions with nice properties so that it can support range query by a depth-first search on the partitioned space.  Here we outline a general scheme for a partition tree which is mostly based on the seminal work by \citet{Chazelle1989}.  
Readers may refer to \cref{fig:tree} for intuition. Those already familiar with the partition tree may skip the discussion following~\cref{def:partitiontree}.

\begin{Definition}[Partition tree scheme]\label{def:partitiontree}
    Given a set of $n$ points $\mathcal{X}$, we preprocess a family of canonical subsets of $\mathcal{X}$ denoted as $\mathcal{N}\subset 2^\mathcal{X}$ and store the weights of those sets in a rooted tree $\mathcal{T} = (\mathcal{V},\mathcal{E})$.  Each node $v\in \mathcal{V}$ of $\mathcal{T}$ is associated with a canonical subset $N(v)\in \mathcal{N}$ called a \emph{node-set} and a list of its children $\mathcal{C}(v)\subset \mathcal{V}$.  
    For any internal node $v$, its children's node-sets form a partition of its node-set so that $\cup_{u\in \mathcal{C}(v)} N(u) = N(v)$ and $N(u)\cap N(u') = \emptyset$ for all $u\neq u'\in \mathcal{C}(v)$.
    The node-set of the root is the universe $N(root) = \mathcal{X}$, and node-sets of leaves are singletons.  
    In additional to the weights $W(N(v))$, each node can store additional auxiliary information, e.g., an encoding of the node-set $N(v)$.\looseness=-1
\end{Definition}

To avoid redundancy, we prohibit any node from having exactly one child, and thus the number of nodes in a partition tree is linear in $n$.  Additionally, common node-set can be encoded succinctly, e.g., the boundary of an interval, and the resulting partition trees are linear-sized data structure. 

Now we illustrate how to use a partition tree to support range query problem and potential issues for range update.  Given a range query with $E\in \mathcal{F}$, we can performs a depth-first search on a partition tree
$\mathcal{T} = (\mathcal{V},\mathcal{E})$ starting from its root with $ans = 0$.  At each node $v\in \mathcal{V}$, we update $ans$ according to following three cases between sets $N(v)$ and $E$.  

\begin{enumerate}
    \item If $E$ contains $N(v)$, we add weight of $N(v)$ to the answer $ans \gets ans + W(N(v))$ and return.
    \item If $E$ and $N(v)$ are disjoint, then return.
    \item If $E$ \emph{crosses} node-set $N(v)$ so that $E$ intersects but does not contain $N(v)$, we recursively call this procedure on all children of $v$.
\end{enumerate}
The above procedure partitions range $E$ into a collection of canonical subset from the first case and returns the sum of the weights.  The query time depends on how many nodes are visited and how long it takes to decide the relationship between $N(v)$ and $E$.  We focus exclusively on the number of nodes visited when answering a query and will sweep the latter under the rug.\footnote{This is known as arithmetic model of computation, where attention is focused on the number of arithmetic operations needed to answer a query and not on the number of steps taken by the algorithm.  Contrarily, deciding the relationship of two subsets $E$ and $E'$ can be computationally hard.  Consider two Turing machines: let $E$ represent the set of inputs with length at most $\log n$ where the first Turing machine halts, and $E'$ represent the inputs where the second Turing machine halts.}

The number of nodes visited in a query depends on the complexity of $\mathcal{F}$ (and $\mathcal{N}$).  As the third case, we say $E\subseteq \mathcal{X}$ crosses $A\subseteq \mathcal{X}$ if $E$ intersects but does not contain $A$, and a range query with $E$ \emph{visits} node $v$ if $v$ is the root or $E$ crosses the canonical set of $v$'s parent and the \defn{visiting number} of a partition tree on set system $(\mathcal{X}, \mathcal{F})$ is the maximum number of nodes visited by any single query in $\mathcal{F}$.  The visiting number amounts to the query time and depends on the complexity of $\mathcal{F}$.  For instance, if $\mathcal{F}$ is the power set of $\mathcal{X}$ which consists of all possible subsets of $\mathcal{X}$, the visiting number can be linear in the size of the partition tree by \citet{Chazelle1989} which is of order $|\mathcal{X}| = n$.  On the other hand, we may design optimal partition trees minimizing the visiting number as long as the set system does not allow queries to cross $\mathcal{X}$ in a fairly arbitrary manner.  

However, the update operation can be more expensive when the range $E$ is large.  For instance, if $E = \mathcal{X}$, the update affects all canonical sets in the partition tree, which is at least $n$.  One of our contributions is to design a lazy propagation algorithm so that the update time is similar to the above query time with little overhead.

\subsubsection{Lazy propagation on partition trees}\label{sec:lazy}
\Cref{alg:lazyUpdate,alg:lazyquery} present our lazy propagation algorithm for weight update and query respectively.  With \cref{thm:equiv}, our partition-tree-based algorithm can support LMSR for any set system.  \Cref{thm:visiting} shows that not only query time but also update time are big $O$ of the visiting number of the partition tree on the set system.  We discuss the construction of partition trees with small visiting numbers in \cref{sec:examples}.  

\begin{theorem}\label{thm:visiting}
    Given a set system $(\mathcal{X}, \mathcal{F})$ and a partition tree $\mathcal{T}$, the query time $T_Q(n)$ of \cref{alg:lazyquery} and the update $T_U(n)$ of \cref{alg:lazyUpdate} on $\mathcal{T}$ are big $O$ of the visiting number of $\mathcal{T}$ on $(\mathcal{X}, \mathcal{F})$. 

    Moreover, with \cref{thm:equiv}, $\mathcal{T}$ can support price, cost, and buy operations for LMSR with running time big $O$ of the visiting number of $\mathcal{T}$ on $(\mathcal{X}, \mathcal{F})$.
\end{theorem}


We defer the proof and formal algorithm statement (\cref{alg:lazyquery,alg:lazyUpdate}) to the appendix.  We illustrate the main idea of lazy propagation.  Instead of performing the update operation immediately, the lazy propagation technique does the update on demand.  Recall that a node in a partition tree stores or represents the results of a query for the node-set.  If the node-set is contained by the update operation range $E$, then all descendants of the node must also be updated, which results in an undesirable update time.   With lazy propagation idea, in the update algorithm (\cref{alg:lazyUpdate}) we stop our update once the node-set is contained by $E$ and postpone updates to its children by storing this update information in a new variable $\pend$ called lazy value.  A value one in $\pend(v)$ indicates that there are no pending updates on node $v$.  A non-identity value means that all descendants need to be multiplied by this amount before making any query to the node.  Since we postpone some updates, we also need to modify our query algorithm (\cref{alg:lazyquery}).  Our algorithm first updates the node if there is a pending update and pushes the lazy value to its children.  Once it makes sure that the pending update is done, it works the same as the original query function.

In \cref{app:rqrug}, we extend those algorithms to  general RQRU with general query and update functions.

\subsubsection{Applications of partition-tree-based algorithms}\label{sec:examples}
We outline various approaches to construct a partition tree with small visiting numbers and summarize our results in \cref{tab:results}.  Many of these outcomes stem from leveraging existing research in computational geometry, with additional examples available in surveys~\cite{har2011geometric,toth2017handbook}.  
First, several set systems already have optimal partition trees, including intervals and orthogonal set systems (\cref{ex:one,ex:orthogonal}).
Second, the dual shattering dimensions of set systems (\cref{def:shattering}) provide tight bounds on the optimal visiting numbers, and admit algorithms to construct near-optimal partition trees for any set system.  Lastly, we show that our algorithm achieves sublinear running time for LMSR on any finite VC dimensional set systems with polynomial construction time.


For one dimensional interval securities (\cref{ex:one}), the simple balanced binary trees (\cref{fig:tree}) have visiting numbers in $O(\log n)$, and thus support $O(\log n)$ running time using \cref{thm:visiting}.  
\begin{corollary}\label{prop:one}
    \Cref{alg:lazyquery,alg:lazyUpdate} with the partition tree in \cref{fig:tree} can support LMSR on one dimensional intervals (\cref{ex:one}) in $O(\log n)$.  Additionally, the partition tree uses linear space and can be constructed in $O(n)$.  
\end{corollary}

\begin{figure}
    \centering
    \includegraphics[width = \textwidth]{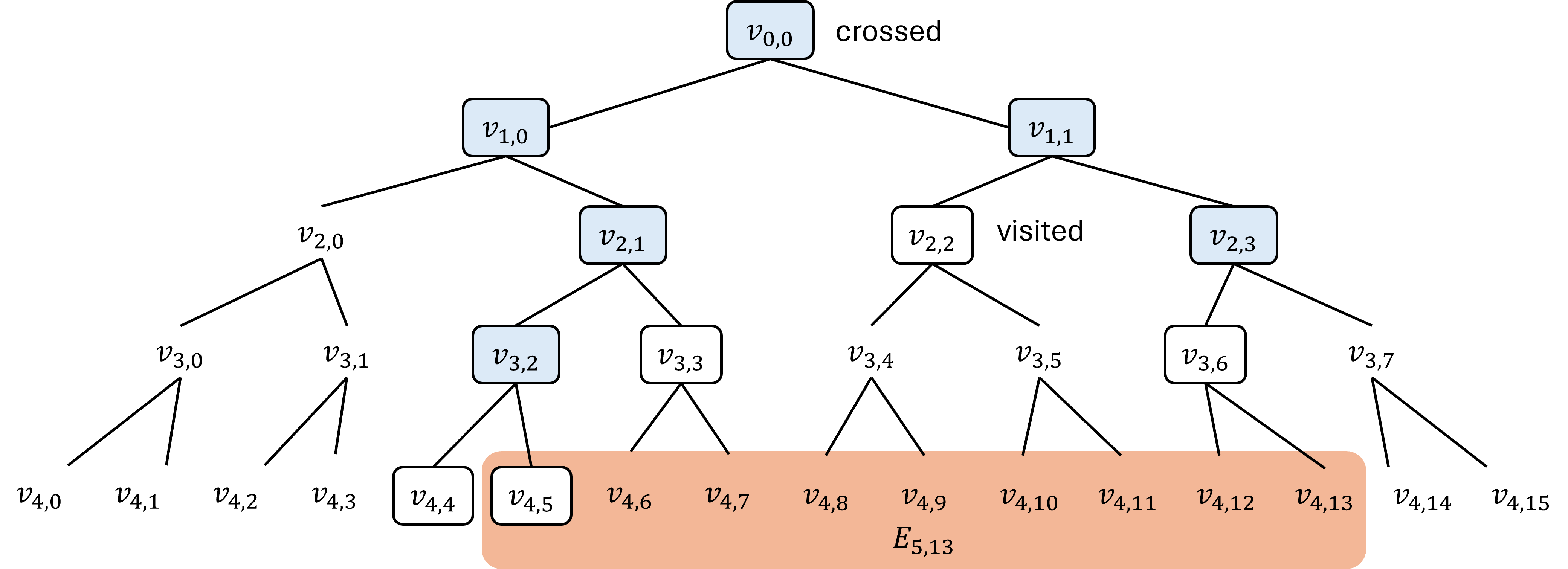}
    \caption{A partition tree for \cref{ex:one} with $n = 16$.  In the figure, we consider a range query with $E_{5,13} = \{5,6,\dots, 13\}$, the squared nodes are visited by the query and blue ones has node set crossed by $E_{(5,13)}$.  More generally, given $n$ points $\mathcal{X} = \{0, \dots, n-1\}$, let $K= \lceil \log_2(n)\rceil$, the height of partition tree is $K$ with node $\mathcal{V} = \{v_{k,l}: k = 0,\dots, \lceil \log_2(n)\rceil, l = 0,\dots, 2^k-1\}$ where $v_{k,l}$ is the $l$-th node at the $k$-th level with node-set $\mathcal{N}(v_{k,l}) = \{l2^{K-k}, l2^{K-k}+1, \dots, (l+1)2^{K-k}-1\}\cap \mathcal{X}$.}
    \label{fig:tree}
\end{figure}
For $d$-dimensional orthogonal securities (\cref{ex:orthogonal}), the $k$-d trees~\cite{DBLP:books/lib/Berg97} have visiting numbers in $O(n^{1-1/d})$. 
\begin{corollary}\label{prop:orthogonal}
    For all $d\ge 2$, \cref{alg:lazyquery,alg:lazyUpdate} with the $k$-d  trees can support LMSR for $d$-dimensional orthogonal set system (\cref{ex:orthogonal}) in $O(n^{1-1/d})$.
    Moreover, the $k$-d trees use linear space and can be constructed in $O(n\log n)$. 
\end{corollary}
For hyperplanes (\cref{ex:hyperplane}), the optimal partition trees by \citet{chan2010optimal} have visiting numbers in $O(n^{1-1/d})$.  
\begin{corollary}\label{prop:hyperplance}
For all $d\ge 2$, \cref{alg:lazyquery,alg:lazyUpdate} with the partition trees by \citet{chan2010optimal} can support LMSR for $d$-dimensional hyperplane set system (\cref{ex:hyperplane}) in $O(n^{1-1/d})$.  Moreover, the partition tree trees use linear space and can be constructed in $O(n\log n)$.
\end{corollary}

More generally, the optimal visiting number can be bounded by the {dual shattering dimension} (\cref{def:shattering}) of the set system~\cite{haussler1986epsilon,Chazelle1989}.
The following result combines \citet{Chazelle1989}'s reduction from low crossing spanning trees to small visiting number partition trees and \citet{DBLP:conf/compgeom/CsikosM21}'s randomized algorithm for low crossing spanning trees.
\begin{theorem}[\cite{Chazelle1989,DBLP:conf/compgeom/CsikosM21}]\label{thm:construct}
Given a set system $(\mathcal{X}, \mathcal{F})$ on $n$ points with dual shattering dimension $d\ge 1$, we can construct a linear-sized binary balanced partition tree so that the expected visiting number is in $O(n^{1-1/d}\log n+ \log|\mathcal{F}|(\log n)^2)$ with an expected $\tilde{O}(|\mathcal{F}|n^{2/d}+n^{2+2/d})$ calls to the membership oracle of $(\mathcal{X}, \mathcal{F})$ that decides whether a given point is a range.

Moreover, there does not exists a partition tree with visiting number of $o(n^{1-1/d})$.
\end{theorem}

The lower bounds on the visiting numbers in \cref{thm:construct} not only show the algorithmic results in \cref{thm:construct} is tight up to poly log factor, but also show the visiting numbers in \cref{prop:one,prop:hyperplance,prop:orthogonal} are optimal.  Since the dual shattering dimension is bounded if and only if the VC dimension is bounded, combining \cref{thm:visiting,thm:construct}, we can have a sublinear time LMSR algorithm on any finite VC dimensional set system. 
\begin{corollary}\label{prop:abstract}
    Given a set system $(\mathcal{X}, \mathcal{F})$ with $|\mathcal{X}| = n$, if the set system has a finite VC dimension $D$ and a membership oracle, there exists a constant $\epsilon_D>0$ so that \cref{alg:lazyquery,alg:lazyUpdate} with \cref{thm:equiv} can support LMSR on $(\mathcal{X}, \mathcal{F})$ with running time in $O(n^{1-\epsilon_D})$ and an expected ${O}(poly(n))$ calls to the membership oracle.
\end{corollary}


\subsection{Hardness results for LMSR}\label{sec:hardness}
The equivalence in \cref{thm:equiv} not only enables efficient LMSR algorithms as discussed in the previous section but also provides a venue for hardness results.  Below, we list several hardness results by reducing existing classical problems to LMSR algorithm problems.

First, we can use a classical hardness result on \emph{dynamical partial sum problem}~\cite{patrascu2004tight} to show that there is no $o(\log n)$ time LMSR algorithm on the one-dimensional intervals (\cref{ex:one}) which implies that the LMSR algorithms in \citet{Dudik21} and \cref{prop:one} are optimal.  We defer the proof to the appendix.
\begin{corollary}\label{prop:hard1d}
The running time of any LMSR algorithm on one dimensional  intervals with $\mathcal{X} = [n]$ (\cref{ex:one}) is in $\Omega(\log n)$.     
\end{corollary}

Second, we can reduce matrix multiplications to LMSR algorithm on $d$-dimensional interval securities when $d\ge 2$ (\cref{ex:orthogonal}).  Consequently, a sub-polynomial time\footnote{$T_P(n), T_B(n)$ are in $o(n^c)$ for all $c>0$.} LMSR algorithm for $d$-dimensional orthogonal set system will solve matrix multiplication in near quadratic time. This connection underscores a significant challenge, considering the current leading algorithm for $m$-by-$m$ matrix multiplication requires $O(m^{2.371552})$.~\cite{williams2023new}  We defer the proof to the appendix.
\begin{proposition}\label{prop:hard2d}
If an LMSR algorithm on $2$-dimensional regular orthogonal set system with $\mathcal{X} = \{(i,j): i, j\in [m]\}$ in \cref{ex:orthogonal} can support price operation in $T_P(m^2)$ and buy operation in $T_B(m^2)$ with $O(m^2)$ preprocessing time, we can solve matrix multiplication in $O(m^2 (T_P(m^2)+T_B(m^2)))$.     
\end{proposition}

Finally, \citet{Chazelle1989} show that if the VC-dimension of $(\mathcal{X}, \mathcal{F})$ is infinite, there is no sublinear time algorithm for range query using linear space.  We can again apply \cref{thm:equiv} and have the following. 
\begin{proposition}\label{prop:limit_vc}
    If the VC-dimension of $(\mathcal{X}, \mathcal{F})$ is infinite, there is no sublinear time LMSR algorithm on $(\mathcal{X}, \mathcal{F})$ using linear space. 
\end{proposition}

Using the above results, we can show that the pairing and $1$-junta securities in \cref{ex:pairing,ex:boolean} cannot have a sublinear time LMSR algorithm by showing the VC dimensions are infinite.

\begin{corollary}\label{cor:vc}
    Given a positive integer $K$, there is no LMSR algorithm on pairing securities that has running time in $o(K!)$ and uses $O(K!)$ space.  Similarly, there is no LMSR on $1$-junta securities that has running time in $o(2^K)$ and uses $O(2^K)$ space.
\end{corollary}

Moreover, as the disjunctions of two coordinate securities in \citet{DBLP:journals/corr/abs-0802-1362} contains $1$-junta securities as special cases, the lower bound in \cref{cor:vc} applies.  Specifically,  there is no $o(2^K)$ time LMSR using linear space. It's worth noting that because \#P is in EXP.  If \#P $\subsetneq$ EXP, this result is stronger than the \#P hard result in \citet{DBLP:journals/corr/abs-0802-1362}.  We defer the proof to the appendix.
\section{Beyond LMSR}\label{sec:beyondlmsr}
So far, we've explored the reduction of the LMSR  problem to range query problems and employed the partition tree scheme to solve the LMSR.  In this section, we extends our framework to design market makers for other scoring rules by mapping them to specific range query problems and leveraging the partition tree framework.  

First, note that not all cost-function-based market makers admit to efficient price, cost, and buy operations.  In \cref{app:hardness}, we construct a cost function that is convex, differentiable, and $\mathbf{1}$-invariant but is NP-hard to compute.  Therefore, we cannot answer price, cost, and buy operations in polynomial time unless NP = P.  Below however, we demonstrate that our framework applies to several commonly used cost-function-based market maker.

We design market maker algorithms for the quadratic scoring rule (QMSR), another widely-used proper scoring rule.  While LMSR can be formulated as a range query with \emph{multiplication range updates} (\cref{def:rqru}), interestingly, in \cref{sec:qmsr}, we show that QMSR can be formulated as a range query with \emph{addition range updates}.  This observation enables us to apply our lazy update algorithm on the partition tree to solve QMSR with the same computational complexity as \cref{thm:visiting}.  Moreover, unlike the hardness result for subpolynomial time LMSR algorithms on orthogonal set systems in \cref{prop:hard2d}, we present a polylogarithmic time algorithm for QMSR using recent advancements in range query with addition range updates~\cite{Ibtehaz_2021}.

To show the generality of our framework, we further study market makers for power scoring rules~\cite{dawid2007geometry,jose2008scoring}, and show our lazy update algorithm remains applicable when using a power scoring rule with degree $\frac32$.  The key observation is to reduce the market maker problem as a range query with \emph{group action range updates}, which contain multiplication and addition updates as special cases.

Finally, we show that the multi-resolution market design can be naturally integrated into the partition-tree scheme in  \cref{sec:multi-res}.  We demonstrate that with efficient and local weight updates, such multi-resolution design will not affect our characterization of market complexity.
\subsection{Quadratic scoring rule market maker}\label{sec:qmsr}
By \cref{def:generalcost} and \cref{eq:QMSR}, the AMMs for prediction markets with QMSR are defined as the following.

\begin{Definition}\label{def:qmsr}
    Given a set system $(\mathcal{X}, \mathcal{F})$, a {QMSR on $(\mathcal{X}, \mathcal{F})$} taking an initial state $\vw_0:\mathcal{X}\to \mathbb{R}$ and offering securities for all $E\in \mathcal{F}$, supports a sequence of operations taking one of the following forms: for any set $E\in \mathcal{F}$, shares $s\in \mathbb{R}$, and state $\vw$,
\begin{itemize}
\item $\price_Q(E; \vw)$: return the current price of security for $E$ where $$
    \price_Q(E;\vw) := \sum_{x\in E} \frac{\partial C_Q(\vw)}{\partial w_x} = \frac{1}{n}|E|+\frac{1}{2b}\sum_{x\in E}w_x-\frac{|E|}{2bn}\left(\sum_{x\in \mathcal{X}}w_x\right).$$
    \item $\cost_Q(E, s; \vw)$: return the current cost of $s$ shares of $x$,         $C_Q(w+s\mathbf{1}_E)-C_Q(w)$
    \item $\buy_Q(E, s;\vw)$: update the state $w(x)\gets w(x)+s$ for all $x\in E$ and $w(x') \gets w(x')$ for $x'\notin E$. 
\end{itemize}
\end{Definition}
Similar to LMSR, we can reduce the QMSR as a new RQRU problem (\cref{def:rqruq}) which replaces multiplication updates in \cref{def:rqru} with addition updates.  
\begin{Definition}\label{def:rqruq}
    Given a positive integer $l$, a set system $(\mathcal{X}, \mathcal{F})$ and initial weights $Z_0:\mathcal{X}\to \mathbb{R}^l$, the \defn{range query with addition range update}, $(+,+)$-RQRU ro short, requests a sequence of operations, taking one of the following forms: for any $E\in \mathcal{F}$ and $S\in \mathbb{R}^l$:
\begin{itemize}
    \item $\query_Q(E; Z)$: compute and return the sum of weights in range $E$, $Z(E) = \sum_{x \in E} Z(x)\in \R^l$.
    \item $\update_Q(E, S; Z)$: update $Z(x) \gets S+Z(x), \forall x\in E$, and $Z(x') \gets Z(x'), \forall x'\notin E$.  
\end{itemize}
\end{Definition}    
In contrast to \cref{def:rqru} and \cref{def:lmsr}, \cref{def:rqruq} maintains vector-valued weights and does not necessary have a straightforward relation to the state in \cref{def:qmsr}.  Thus, we use $Z$ instead of $W$ to highlight this distinction.
Similar to the first part of \cref{thm:equiv}, we show QMSR market in \cref{def:qmsr} can be reduced to $(+,+)$-RQRU problem in \cref{def:rqruq}.  We defer the proof to the appendix.

\begin{lemma}\label{lem:equivq}
Given a set system $(\mathcal{X}, \mathcal{F})$, if there is a $(+,+)$-RQRU algorithm on $(\mathcal{X}, \mathcal{F})$ defined in \cref{def:rqruq} with $T_Q(n)$ query time and $T_U(n)$ range update time, there exists a QMSR market on $(\mathcal{X}, \mathcal{F})$ that can support price operation in $O(T_Q(n))$, buy operation in $O(T_U(n)+T_Q(n))$, and cost operation in $O(T_Q(n))$. 
\end{lemma}

Thus, the partition tree scheme in \cref{alg:lazyquery,alg:lazyUpdate} adapts to $(+,+)$-RQRU, and thus shares the same complexity as \cref{thm:visiting}.
\begin{theorem}\label{thm:qmsr}
Given a set system $(\mathcal{X}, \mathcal{F})$ and a partition tree $\mathcal{T}$, a lazy propagation algorithm on $\mathcal{T}$ can support price, buy, and cost operation for QMSR with running time big $O$ of the visiting number of $\mathcal{T}$ on $(\mathcal{X}, \mathcal{F})$.
\end{theorem}

Moreover, recent paper~\cite{Ibtehaz_2021} proposes multi-dimensional segment tree for addition range updates on regular $d$-dimensional orthogonal set with $\mathcal{X} = [m]^d$ (\cref{ex:orthogonal}).  By \cref{lem:equivq}, we can consequently achieve a QMSR algorithm with polylogarithmic time complexity, in contrast to the subpolynomial time hardness results for LMSR presented in \cref{prop:hard2d}.  
\begin{corollary}\label{cor:orthogonalp}
For all $d\ge 2$, $m$, and $n = m^d$, there is multi dimensional segment tree that support QMSR for regular $d$-dimensional orthogonal set system with $\mathcal{X} = [m]^d$ in $O(\log^d(n))$.
\end{corollary}


\subsection{Power scoring rule market maker}\label{sec:pmsr}
Finally, we study automated market maker mechanisms for $\gamma$-power scoring rules with the cost function defined in~\cref{eq:pmsr}.
\begin{Definition}\label{def:pmsr}
    Given a set system $(\mathcal{X}, \mathcal{F})$, a {$\gamma$-power MSR on $(\mathcal{X}, \mathcal{F})$} taking an initial state $\vw_0:\mathcal{X}\to \mathbb{R}$ and offering securities for all $E\in \mathcal{F}$, supports a sequence of operations taking one of the following forms: for any set $E\in \mathcal{F}$, shares $s\in \mathbb{R}$, and state $\vw$,
\begin{itemize}
\item $\price_\gamma(E; \vw)$: return the current price of security for $E$, $\sum_{x\in E} \frac{\partial C_\gamma(\vw)}{\partial w_x}$. 
    \item $\cost_\gamma(E, s; \vw)$: return the current cost of $s$ shares of security for $E$,   $C_\gamma(\vw+s\mathbf{1}_E)-C_\gamma(\vw)$.
    \item $\buy_\gamma(E, s;\vw)$: update the state $w(x)\gets w(x)+s$ for all $x\in E$ and $w(x') \gets w(x')$ for $x'\notin E$. 
\end{itemize}
\end{Definition}

We provide AMMs for $\gamma = \frac32$ which admits the following closed-form solution, and we defer the proof to the appendix.

\begin{lemma}\label{prop:pmsr_closedform}
    Given a state $\vw:\mathcal{X}\to \R$, let $M_1 = \sum_{x\in \mathcal{X}} w_x$, $M_2 = \sum_{x\in \mathcal{X}} w_x^2$, $M_3 = \sum_{x\in \mathcal{X}} w_x^3$, and $\mu = \sqrt{M_1^2-n\left(M_2-\frac{9b^2}{4}\right)}$.  The cost function of $\frac32$-power MSR is $$C_{\frac32}(\vw) = \max_{p\in \Delta_\mathcal{X}} \sum_{x\in \mathcal{X}} w_xp(x)-b\sum_{x\in \mathcal{X}}p(x)^{\frac32} = \frac{4}{27b^2}\left( M_3-\frac{1}{n^2}\left(M_1^3-3M_1^2\mu+2\mu^3 \right)\right)$$
    and price function is
    $$\frac{\partial}{\partial w_x}C_{\frac32}(\vw) = \frac{1}{n}+\frac{4}{9}\left(w_x^2+\frac{2}{n}(\mu-M_1) w_x-\frac{1}{n}M_2+\frac{2}{n^2}M_1^2-\frac{2}{n^2}M_1\mu\right).$$
\end{lemma}

With the above closed form, we reduce $\frac32$-power MSR problem as a range query and range update problem.

\begin{Definition}\label{def:rqrup}
    Given a set system $(\mathcal{X}, \mathcal{F})$ and initial weights $Z_0:\mathcal{X}\to \mathbb{R}^4$, the range query range update problem for $\frac32$-power MSR, called $(+, \alpha)$-RQRU for short,  requests a sequence of operations, taking one of the following forms: for any $E\in \mathcal{F}$ and $S\in \mathbb{R}$:
\begin{itemize}
    \item $\query_{\frac32}(E; Z)$: compute and return the sum of weights in range $E$, $Z(E) = \sum_{x \in E} Z(x)\in \R^d$.
    \item $\update_{\frac32}(E, S; Z)$: for each $x \in E$, update $Z(x) \gets \alpha_S(Z(x))$, and for each $x'\notin E$, $Z(x') \gets Z(x')$ where 
    $$\alpha_S\left(\begin{bmatrix} \Gamma_0\\ \Gamma_1\\\Gamma_2\\\Gamma_3\end{bmatrix}\right) = \begin{bmatrix} \Gamma_0\\ \Gamma_1+S\\\Gamma_2+2S\Gamma_1+S^2\\\Gamma_3+3S\Gamma_2+3S^2\Gamma_1+S^3\end{bmatrix}\text{for all}\begin{bmatrix} \Gamma_0\\ \Gamma_1\\\Gamma_2\\\Gamma_3\end{bmatrix}\in \R^4.$$
\end{itemize}
\end{Definition}    
Despite the complicated appearance, $(+, \alpha)$-RQRU problem in \cref{def:rqrup} essentially maintains the sum of powers from degree $0$ to $3$  ($M_0, M_1, M_2$ and $M_3$ in \cref{prop:pmsr_closedform}) which is sufficient to compute the price and cost function of $\frac32$-power MSR.  This connection is formally outlined in \cref{lem:equivp}.  We defer the proof to the appendix.


\begin{lemma}\label{lem:equivp}
Given a set system $(\mathcal{X}, \mathcal{F})$, if there is a $(+,\alpha)$-RQRU algorithm on $(\mathcal{X}, \mathcal{F})$ defined in \cref{def:rqrup} with $T_Q(n)$ query time and $T_U(n)$ range update time, there exists a $\frac32$-power MSR market on $(\mathcal{X}, \mathcal{F})$ that can support price operation in $O(T_Q(n))$, buy operation in $O(T_U(n)+T_Q(n))$, and cost operation in $O(T_U(n)+T_Q(n))$. 
\end{lemma}

Finally, the partition tree scheme in \cref{alg:lazyquery,alg:lazyUpdate} adapts to this new RQRU problem and shares the same complexity as \cref{thm:visiting,thm:qmsr}.

\begin{theorem}\label{thm:pmsr}
Given a set system $(\mathcal{X}, \mathcal{F})$ and a partition tree $\mathcal{T}$, a lazy propagation algorithm on $\mathcal{T}$ can support price, buy, and cost operation for $\frac32$-power MSR (\cref{def:pmsr}) with running time big $O$ of the visiting number of $\mathcal{T}$ on $(\mathcal{X}, \mathcal{F})$.
\end{theorem}


\begin{Remark}
Though \cref{thm:qmsr,thm:pmsr} only apply to $\gamma$-power MSR with $\gamma = 2$ and $\frac32$, we believe the framework can be expanded to more general power MSR.  The main challenge is deriving a closed form for the cost function from \cref{eq:pmsr} which is constrained convex conjugate of the scoring rule.  However, it is well-known that the unconstrained convex conjugate of a polynomial with degree $\gamma$ is a polynomial of degree $\frac{\gamma}{\gamma-1}$.  Such relationship is may still hold in \cref{eq:pmsr} as the cost function of quadratic scoring rule (\cref{eq:QMSR}) depends on a polynomials of degree $\frac{2}{2-1} = 2$, and the cost function of $\frac32$ power scoring rule (\cref{prop:pmsr_closedform}) depends on polynomials of degree $\frac{\frac32}{\frac32-1} = 3$.  Our partition tree method can easily maintain the sum of the power of a higher degree and thus has the potential to solve more general $\gamma$-power MSR.  Finally, using Taylor approximation, we should be able to provide approximated market maker for general scoring rules.
\end{Remark}
In \cref{app:rqrug} we show that multiplication (\cref{def:rqru}) and the above $\alpha_S$ function update (\cref{def:rqrup}) are special cases of group action updates where the our partition tree method remains applicable. 
\subsection{Partition tree and multi-resolution market}
\label{sec:multi-res}

In this section, we demonstrate that \emph{a multi-resolution market} can be naturally combined with the partition-tree scheme.
The multi-resolution design grants the flexibility to adopt distinct scoring rules, when aggregating information at different resolutions.
We show that with \emph{efficient and local weight updates}, such multi-resolution design will not affect our characterization of market complexity for all operations, including the removal of arbitrage opportunities that arise from AMMs using different market scoring rules for combinatorial securities associated with information at different granularity.

\begin{Definition}[An independent multi-resolution market]
    \label{def:multiresolution_market}
    A multi-resolution market with $(\mathcal{N}_k, C_k)_{k = 0,\dots, K}$ consists of $K$ submarkets on $\mathcal{X}$.  Each submarket $k =0,\dots, K$ uses a cost function $C_k$ and offers combinatorial securities associated with a set system $\mathcal{N}_k$, where $\mathcal{N}_0 = \{\mathcal{X}\}$ and $\mathcal{N}_k$ forms a partition of $\mathcal{X}$ that is finer than $\mathcal{N}_{k-1}$. 
    
    Moreover, a multi-resolution market is a \emph{consistent multi-resolution market} if there exists a cost function $C: \mathcal{X}\to \R$ and a sequence of liquidity parameters $\vb = (b_0, \dots, b_K)$ such that $C_k(w) = \frac{1}{b_k}C\left(\frac{w}{b_k}\right)$ for all $k$ and $w$. 
\end{Definition}

Note that we can reduce a multi-resolution market to a partition tree as the following.%
\footnote{Conversely, given a partition tree of depth $K$ where every leaves are at the same level, we can define a multi-resolution market where submarket $k$ offers combinatorial securities for the node-sets at level $k$.}  
Let canonical subsets be $\mathcal{N} = \cup_{k = 0}^K \mathcal{N}_k$.  
We construct a partition tree $\mathcal{T} = (\mathcal{V},\mathcal{E})$ of depth $K$ on $|\mathcal{N}|$ nodes: 
let $\mathcal{V}_k$ for $k \in \{0,1,\dots,K\}$ be the set of nodes at each level of $\mathcal{T}$.  Each node $v\in \mathcal{V}_k$ is associated with a node-set $N(v)$ in $\mathcal{N}_k$, and has a list of children $\mathcal{C}(v)$ whose associated node-set $N(u)\in \mathcal{N}_{k+1}$ and $N(u)\subset N(v)$ for all $u\in \mathcal{C}(v)$.  Therefore, we can reuse the notion of visiting number in \cref{sec:pre_partition} to measure the complexity of a multi-resolution market.  Formally, the \emph{visiting number of a multi-resolution market} is the visiting number of the constructed partition tree on the set system $(\mathcal{X}, \cup_{k = 0}^K \mathcal{N}_k)$.

    

\begin{example}[A multi-resolution Gates Hillman prediction market]
\label{ex:multiresolution}
The Gates Hillman prediction market (GHPM) was designed to predict the opening day of the Gates and Hillman Centers at Carnegie Mellon University~\cite{Othman10}.
A multi-resolution variation of such a market can contain its \emph{quarter submarket} (trading securities to bet on during which quarter the center will open), \emph{month submarket}, \emph{week submarket}, and \emph{day submarket}, with each having their distinct market scoring rule to facilitate aggregating information at different granularity. 
\end{example}

This additional flexibility in designing each submarket enable the designer to allocate budget and choose $C_k$ that reflects the ``granularity'' of a security (e.g., smaller liquidity parameters for submarkets with more complex or fine-grained securities), in effect facilitating information elicitation and price convergence~\cite{Dudik21}.
%
However, under such multi-resolution construction, as any canonical set $N(v)$ in a coarser market (e.g., quarter market) can be also expressed in a finer one (e.g., month market), running submarkets independently may lead to incoherent prices and introduce arbitrage opportunities.

To maintain price coherence, we follow \citeauthor{Dudik21} in designing a \emph{linearly constrained market maker} (multi-resolution LCMM)~\cite{DudikLaPe12,Dudik21}.
It imposes linear constraints to tie market prices among different submarkets and to remove any arbitrage opportunity. 
Let $\M$ denotes a \emph{coherent price space}.
For the multi-resolution market constructed on $\mathcal{T}$, we use constraint matrix $\A$ to specify a set of \emph{homogeneous linear equalities} that describe a superset of~$\M$:
\begin{equation}
\label{eq:Amu=0}
\M\subseteq\{\vmu\in\R^{\card{\mathcal{V}}}:\:\A^\top\vmu=\vzero\}.
\end{equation}
Arbitrage opportunities arise whenever prices fall outside the set of coherent prices $\M$~\cite{AbernethyChVa11}.\fang{It seems that we do not need to specify $\mA$.  Instead we require $\mu_u = \sum_{v\in\V_k:\:v\subset u} \mu_v, 
\text{ for any $l<k\leq K$.}$}
We generalize to the partition-tree scheme and define the constraint matrix $\A$ to ensure that $\mu(N(v)) = \sum_{u \in \mathcal{C}(v)} \mu(N(u))$, for any node $v$. 
Let $\U=\V \backslash \V_K$ be the set of inner nodes of $\mathcal{T}$, the matrix $\A\in\R^{\card{\V}\times\card{\U}}$ can be defined as:
\begin{equation}
\label{eq:A_gen}
A_{vu}=
\begin{cases}
1 & \text{if $v=u$,}
\\
-1 & \text{if $v\in \mathcal{C}(u)$,%
}
\\
0 & \text{otherwise.}
\end{cases}
\end{equation}
We verify that $\A$ enforces price coherence in the proof of \cref{prop:multi-resolution-cost}, i.e., for each inner node $u$ at level $\ell = \level(u)$, we have
$
\mu_u = \sum_{v\in\V_k:\:v\subset u} \mu_v, 
\text{ for any $l<k\leq K$.}
$
\footnote{For simplicity, we here abuse notation and write $v \subset u$ to denote that $v$ is a strict descendant of $u$, i.e., $N(v)\subset N(u)$.}
We note that \cref{eq:A_gen} is just one form of constraint matrices, and the design of $\A$ can be adapted to facilitate local weight updates to remove arbitrage (shown later in \cref{ex:multiresolution_lmsr} and \cref{def:local_update}).

We leverage the defined linear constraints in matrix $\A$ to remove arbitrage. 
We denote the state for each submarket $k$ as $w_k:\mathcal{N}_k \to \mathbb{R}$, and they form the block of coordinates for the overall multi-resolution market state $w$ (i.e., $w$ is the concatenation of $(w_0, w_1, \dots, w_K)$). 
Given a cost function for each submarket $C_k(w_k)$, we have the direct-sum cost $\Cx(w)=\sum_{k\le K} C_k(w_k)$.
The multi-resolution LCMM is then described by the following cost function:
\begin{equation}
\label{eq:lcmm_cost}
C(w) = \inf_{\eta\in\R^{\card{\U}}}\Cx(w+\A\eta).
\end{equation}
To implement the above cost function, we keep track of the state $\tilde{w} = w+\A\eta$ in the direct-sum market $\Cx$.
Specifically, with a trader purchasing $\delta$ that introduces arbitrage opportunities, we update $w$ to $w'=w+\delta$, seek the lowest cost for the trader by buying the corresponding bundles $\A\delta_{arb}$ on the traders behalf to remove arbitrage, and then update $\eta'=\eta+\delta_{arb}$.%
\footnote{We note that the purchase of bundle $\A\delta_{arb}$ has no effect on the trader's payoff, given coherent prices as in \cref{eq:Amu=0}.}
The resulting cost for the trader then is $\Cx(w'+\A\eta')-\Cx(w+\A\eta)$.
\begin{lemma}
    \label{prop:multi-resolution-cost}
    The constraint matrix $\A$ (\cref{eq:A_gen}) enforces price coherence across all submarkets.
    The cost function of the multi-resolution LCMM (\cref{eq:lcmm_cost}) removes any arbitrage opportunity that violates linear constraints specified in matrix $\A$.
\end{lemma}

Below we give an example of one specific form of multi-resolution markets, referred to as \emph{a multi-resolution LMSR market}~\cite{Dudik21} with a variant of the constraint matrix to support local updates.
\begin{example}[A multi-resolution LMSR market]
\label{ex:multiresolution_lmsr}
In a multi-resolution LMSR market, each submarket $k$ adopts the LMSR cost function $C_k$ with a \emph{separate} liquidity parameter $b_k > 0$, i.e., $C_k(\tilde{w}_k)=b_k\ln \big(\sum_{v\in\mathcal{V}_k} e^{\tilde{w}_k(v)/b_k}\big)$.
We can equivalently define the following constraint matrix
$\A^{\text{LMSR}}\in\R^{\card{\V}\times\card{\U}}$ as:
\begin{equation}
\label{eq:A}
A^{\text{LMSR}}_{vu}=
\begin{cases}
B_{\level(v)} & \text{if $v=u$,}
\\
-b_{\level(v)} & \text{if $v \subset u$,}
\\
0 & \text{otherwise,}
\end{cases}
\end{equation}
where $B_{
\level(v)} = \sum_{k=\level(v)+1}^K b_k$.
To verify that $\A^{\text{LMSR}}$ enforces price coherence, for each inner node $u$ at level $\ell = \level(u)$, by induction, we have
$
\big(\sum_{k>\ell} b_{k}\big) \mu_u = 
\sum_{k>\ell} \big(b_k\sum_{v\in\V_k:\: v \subset u} \mu_v\big)$ if and only if $ 
\mu_u = \sum_{v\in\V_k:\:v\subset u} \mu_v 
$, for any $l<k\leq K$.
\end{example}

Given price coherence constraints, the challenge left is to have an efficient, local weight update to remove arbitrage across submarkets, which we formally define below.

\begin{Definition} [Efficient and local arbitrage removal in multi-resolution markets]
    \label{def:local_update}
    Given a designed constraint matrix $\A^*$ that spans the same subspace as $\A$ in \cref{eq:A_gen}, fix a submarket at level $\ell < K$. 
    Let $\tilde{w}$ be the market state in $\tilde{C}$, where the prices $\tilde{p}(\tilde{w})$ are coherent among all finer submarkets at levels $k>\ell$.
    A \emph{local update} satisfies that for any $x \in \mathbb{R}$ and for any node $u$ with $\level(u)\leq \ell$, the prices after buying $x$ shares of bundle $a^*_u$ (the $u$'s column of $\A^*$), i.e., $\tilde{p}(\tilde{w}+xa^*_u)$, remain coherent among all finer submarkets at levels $k>\ell$.
    
    An \emph{efficient and local arbitrage removal} satisfies that there is a closed-form solution of $x^*$, such that the prices after buying $x^*$ shares of $a^*_u$, i.e., $\tilde{p}(\tilde{w}+x^*a^*_u)$, remain coherent among all submarkets $k\geq\ell$, i.e., any arbitrage between the submarket $\ell$ and all submarkets with $k > \ell$ is removed.
\end{Definition}

To leverage efficient and local arbitrage removal, we start with a market state $\tilde{w}$ (e.g., $w=\vzero$ and $\eta=\vzero$), where all submarkets are coherent.
When some shares of security associated with $N(u)$ is traded, the submarket $\ell=\level(u)$ loses price coherence with others.
By buying a closed-form amount $x^*$ of $a_u$ (i.e., $\eta(u)\gets\eta(u)+x^*$), it is possible to restore coherence between $\ell$ and $\ell+1$, and \emph{local update} then implies that coherence with all finer levels $k>\ell+1$ is not disrupted. 
The process of restoring coherence can then go up to the parent of $u$ and the bundle vector $a^*_{\pt(u)}$.
Based on \cref{ex:multiresolution_lmsr}, we further illustrate efficient and local arbitrage removal in multi-resolution LMSR markets as a result of the constructed $\A^{\text{LMSR}}$. 

\begin{example}[Price, cost, and arbitage removal in multi-resolution LMSR markets]
    \label{ex:multiresolution_lmsr_arb}
    Given a coherent multi-resolution LMSR market with $w$ and $\eta$, to calculate price, let $\tilde{w} = w+\A^{\text{LMSR}}\eta$ be the corresponding state in $\Cx$.
    We consider a node $v\neq root$ with $k:=\level(v)$. 
    We denote the siblings of $v$ as $\sib(v)=\mathcal{C}(\pt(v)) \backslash v$. 
    The price of the security associated with $N(v)$ can be recursively calculated as
    \begin{equation}
    \label{eq:multi_LMSR_price}
    \price(N(v))
    =\frac{\exp\BigParens{\frac{w(v)+B_{k}\eta(v)}{b_k}}}
    {\exp\BigParens{\frac{w(v)+B_{k}\eta(v)}{b_k}}
    	+
    \sum_{u \in \sib(v)} \exp\BigParens{\frac{w(u)+B_{k}\eta(u)}{b_k}}} \cdot \price({\pt(v)}).
    \end{equation}
    Therefore, $\price$ can be calculated along the search path.
    Similar to the vanilla LMSR construction, we can define the weights in RQRU for multi-resolution LMSR as the following,
    \begin{equation}
        W_t(N(v)) = \exp\BigParens{\frac{w_t(N(v))+B_{k}\eta_t(N(v))}{b_k}} \text{ for all } v\in \mathcal{V}_k \text{ and } t = 0,1,\dots .
    \end{equation} 
    Based on $\price$, $\cost$ can be conveniently calculated following \cref{eq:def_cost}.
    
    Given a coherent multi-resolution LMSR market, after a trader buys $s$ shares of security associated with $N(u)$ with $\ell := \level(u)$, it suffices to update $\eta_u$ by a closed-form amount to restore price coherence across finer submarkets (i.e., for all $k\geq l$):
    \begin{equation}
        \label{eq:arb_x}
        x^*_u = \frac{b_\ell}{B_{\ell-1}}\ln\Parens{\frac{1-p_u}{p_u}\cdot\frac{p'_u}{1-p'_u}},
    \end{equation}
    where $p_u$ denotes the price of $N(u)$ after $s$ shares are traded in the submarket $\ell$, and $p'_u$ denotes the price of $N(u)$ in all other finer submarkets (i.e., $k > \ell$) that can express the price of $N(u)$, before arbitrage removal, i.e., $p_u \neq p'_u$ due to the trade.
    We defer calculation details to the appendix.
\end{example}

\begin{lemma}\label{lem:lcmm_q}
If $C_Q$ is the cost for quadratic scoring rule in \cref{eq:QMSR}, any consistent multi-resolution market with $C_Q$ and $(b_k, \mathcal{N}_k)_{k = 0,\dots K}$ has an efficient local arbitrage removal algorithm.
\end{lemma}

\begin{Definition}
    \label{def:multiresmm}
    Given a multi-resolution market that takes a (constructed) partition tree $\mathcal{T}=(\mathcal{V}, \mathcal{E})$, a cost function for each submarket $C_k$ with an initial state $w_k = \vzero$ for $k\leq K$ and an arbitrage state $\eta_k=\vzero$ for $k<K$, a designed constraint matrix $\A \in \mathbb{R}^{\card{\mathcal{V}}\times \card{\mathcal{V} \backslash \mathcal{V}_K}}$, and offers securities for all $E\in \mathcal{F}$. 
    For any set $E\in \mathcal{F}$, shares $s\in \mathbb{R}$, and states $w$ and $\eta$, it supports a sequence of price, cost, and buy operations:
\begin{itemize}
    \item $\price_{MR}(E; w, \eta)$: return the price of security for $E$ in the most fine-grained submarket $K$,
    \[\sum_{v\in \mathcal{V}_K: N(v)\in E} \frac{\partial C_K(\tilde{w}_K)}{\partial \tilde{w}_K(v)}.\] 
    \item $\cost_{MR}(E, s; w, \eta)$: return the current cost of $s$ shares of security for $E$, 
    \[\Cx(w + \A\eta+ s\vone_E +\A s_\textup{arb})-\Cx(w + \A\eta),\]
    where $\vone_E$ has an entry of 1 for nodes that form a partition of $E$, i.e., $v \in Z(E)$.
    \item $\buy_{MR}(E, s;w,\eta)$: update the state $w(u)\gets w(u)+s$ for all $u\in Z(E)$ and $w(u') \gets w(u')$ for $u'\notin Z(E)$, and update the arbitrage state $\eta(v)\gets \eta(v)+s_\textup{arb}(v)$ for all $v\in \{Z(E) \cup \pred(u)\}$ for all $u\in Z(E)$. 
\end{itemize}
\end{Definition}

\begin{theorem}\label{thm:multi-res}
    Given a multi-resolution markets $(\mathcal{N}_k, C_k)$ for $k=0,\cdots, K$, with an efficient and local arbitrage removal (\cref{def:local_update}), we can compute price, buy, cost operation for multi-resolution market (\cref{def:multiresmm}) in time big $O$ of the visiting number of the multi-resolution market.
\end{theorem}
Combining ~\cref{ex:multiresolution_lmsr_arb,lem:lcmm_q}, the above theorem implies that the consistent multi-resolution markets with log and quadratic scoring rules have the same computational complexity as LMSR~\cref{def:lmsr} and QMSR~\cref{def:qmsr}.

\section{AMMs for decentralized finance}\label{sec:swap}
A \defn{constant function market maker} (CFMM) for a finite set of $n$ assets $\mathcal{X}$ maintains a \emph{reserve} of available assets $\vw\in \R^\mathcal{X}$ and a \emph{trading function} $\varphi: \R^\mathcal{X}\to \R$ that is concave and increasing so that $\varphi(\vw)> \varphi(\vw')$ if $w_x\ge w'_x$ for all $x$ and $\vw\neq \vw'$.  
Traders propose to trade or exchange one basket of assets $\vr^+$ for another $\vr^-\in \R^\mathcal{X}$, where $\vr^+$ is referred to as the \emph{tender basket} and $\vr^-$ as the \emph{received basket}.  
The CFMM accepts the proposed trade if $\varphi(\vw+\vr^+-\vr^-) =  \varphi(\vw)$ and updates the reserve to $\vw\gets \vw+\vr^+-\vr^-$. 
Some examples used in practice are the following.
\begin{itemize}
    \item Logarithmic trading function~\cite{angeris2023geometry} with parameter $b\in \R_{>0}$ is
    \begin{equation}\label{eq:swap_log}
        \varphi(\vw) = -\sum_{x\in \mathcal{X}} e^{-w_x/b}
    \end{equation}
    \item Constant (weighted) sum market maker is a linear trading function with predetermined, non-negative parameters $\vc = (c_x)_{x\in \mathcal{X}}$: \begin{equation}\label{eq:swap_lin}
        \varphi(\vw) = \sum_{x\in \mathcal{X}} c_x w_x
    \end{equation}
    \item Another choice of trading function is the (weighted) geometric mean with  non-negative parameters $(\gamma_x)_{x\in \mathcal{X}}$: 
    \begin{equation}\label{eq:swap_prod}
        \varphi(\vw) = \prod_{x\in \mathcal{X}} w_x^{\gamma_x}
    \end{equation}
    Examples include Uniswap v2~\cite{zhang2018formal,adamsuniswapv2}, Balancer~\cite{martinelli2019balancer}, and SushiSwap~\cite{zhang2018formal}.  In particular, Uniswap and SushiSwap use $\gamma_x = 1/n$ for all $x$, and are called \textit{constant product market makers}~\cite{angeris2021analysis}.
\end{itemize}
Again, this paper will focus on logarithmic trading function. 

\paragraph{Swap trades for two baskets in CFMM}
One of the most common trade is \emph{swap} that involves only two assets, one that is tendered and one that is received, i.e., $\vr^+$ and $\vr^-$ only have one nonzero entry at $x^+$ and $x^-$ respectively. 
Thus, we have $\vr^+ = s_+\mathbf{1}_{x^+}$ and $\vr^- = s_-\mathbf{1}_{x^-}$, where $s_- \ge 0$ is the quantity of asset $x^-$ the trader wishes to receive in exchange for the quantity $s_+\ge 0$ of asset $x^+$. 
This is referred to as exchanging asset $x^+$ for asset $x^-$.  

A natural generalization of exchanging two assets is exchanging multiples of two baskets where the market maker tenders and receives a multiple of fixed baskets~\cite{angeris2022constant}.
Thus, we have $\vr^+ = s_+ \tilde{\vr}^+$ and $\vr^- = s_- \tilde{\vr}^-$, where $s_+, s_-\ge 0$ scale the fixed baskets $\tilde{\vr}^+$ and $\tilde{\vr}^-$. In this paper, we consider combinatorial baskets where $\tilde{\vr}^+ = \mathbf{1}_{E^+}$ and   $\tilde{\vr}^- = \mathbf{1}_{E^-}$ for some sets $E^+$ and $E^-$.  When $E^+ = \{x^+\}$ and $E^- = \{x^-\}$, this reduces to the above two-asset trade.  Additionally, one may want to support trades on subsets of assets with bounded cardinality, e.g., Balancer can support swap on sets of set up to eight assets. 


Given $\varphi$, reserve $\vw$, and $\tilde{\vr}^+, \tilde{\vr}^-$, the trade acceptance condition is 
\begin{equation}\label{eq:trading2}
    \varphi(\vw+s_+\tilde{\vr}^+-s_-\tilde{\vr}^-) = \varphi(\vw).
\end{equation}
As $\varphi$ is increasing, there is an one-to-one mapping between $s_+$ and $s_-$. 
First, the forward exchange finds the scale of receiving basket $s_-$ for $s_+\tilde{\vr}^+$ that satisfies \cref{eq:trading2}, and the backward exchange finds the scale $s_+$ of rendering basket for $s_-\tilde{\vr}^-$.  

In this section, we ask \emph{when the number of assets $n$ is large and traders exchange assets under the combinatorial basket setting, whether or how can we support the forward and backward exchange function?}
\begin{Definition}\label{def:swap}
    Given a set system $(\mathcal{X}, \mathcal{F})$, a {combinatorial swap market maker with $\varphi$} taking an initial reserves $\vw_0$ and swap for all $E^+, E^-\in \mathcal{F}$, supports a sequence of swap operations taking one of the following forms:
\begin{itemize}
\item $\tradef(E^-, E^+, s_+, \vw)$: return $s$ so that $\varphi(\vw+s_+\mathbf{1}_{E^+}-s\mathbf{1}_{E^-}) = \varphi(\vw)$ and update $\vw\gets \vw +s_+\mathbf{1}_{E^+}-s\mathbf{1}_{E^-}$.
\item $\tradeb(E^-, E^+, s_-, \vw)$: return $s$ so that $\varphi(\vw+s\mathbf{1}_{E^+}-s_-\mathbf{1}_{E^-}) = \varphi(\vw)$ and update $\vw\gets \vw +s\mathbf{1}_{E^+}-s_-\mathbf{1}_{E^-}$.
\end{itemize}
For simplicity, we assume that the sequence of trade always has a feasible $s$ that is bounded by some constant $\lambda$.
\end{Definition}

In this section, we show that the swap operation in \cref{def:swap} can be reduced from a dynamic algorithm problem---range update problems which can be seen as a special case of range query range update by only supporting query on the universe $\mathcal{X}$.

\begin{Definition}\label{def:rqru_swap}
    The \defn{range update problem} on $(\mathcal{X}, \mathcal{F})$ with  $\varphi$ and $+$, and initial weights, denoted as $(\varphi,+)$-RU.  It requests a sequence of range update $\update(E, S; W)$ such that for each $x \in E$, update $W(x) \gets S+W(x)$, and for each $x'\notin E$, $W(x') \gets W(x')$ and return $\varphi(w)$. 
\end{Definition}
As we will see $\varphi$ depends on the choice of trading function. A function $\varphi$ is \emph{decomposable} if for any $\vw \in \R^\mathcal{X}$, $x\in \mathcal{X}$, and $w_x' \in \R$, we can compute $\varphi(\vw)$ from  $w_x, w_x'$ and $\varphi(\vw_{-x}, w_x')$ in constant time where $(\vw_{-x}, w_x')$ is the vector in $\R^{\mathcal{X}}$ obtained by replacing the $x$-coordinate of $\vw$ with $w_x'$. All above three trading function examples are decomposable.



\begin{proposition}\label{thm:equiv_swap}
Given a set system $(\mathcal{X}, \mathcal{F})$ with $\card{\mathcal{X}} = n$ and $\{x^*\}\in \mathcal{F}$ for some $x^*\in \mathcal{X}$, let $\varphi: \R^\mathcal{X} \to \R$ be a decomposable trading function.  
If there is a $(\varphi,+)$-RU algorithm on $(\mathcal{X}, \mathcal{F})$ with $T_U(n)$ range update time, there exists a combinatorial swap market maker on $(\mathcal{X}, \mathcal{F})$ that can support swap operations in $\tilde{O}(T_U(n))$ with additional log factor depending on the input size and using the same order of space.
    
Conversely, if there is a combinatorial swap market maker on $(\mathcal{X}, \mathcal{F})$ that supports both swap operations in $T_S(n)$, there is a $(\varphi,+)$-RU algorithm with ${O}(T_S(n))$ range update time using the same order of space.
\end{proposition}

The main observation is that if we can compute the trading function $\varphi$, we can determine the scales through a binary search. 
Conversely, we can use the decomposable property to maintain the value of $\varphi$ on $n-1$ coordinates and recover the true value.

As a simple corollary, we can see that the above examples correspond to several interesting RU problems summarized in \cref{tab:results2}.  As show in the previous section both $(+,\cdot)$ and $(+,+)$ have efficient algorithms that depend on the complexity of set system $(\mathcal{X}, \mathcal{F})$. Unfortunately, it remains unclear how to apply our partition tree scheme to the geometric mean CFMM.

\begin{corollary}
    Given a set system $(\mathcal{X}, \mathcal{F})$ and a partition tree $\mathcal{T}$, a lazy propagation algorithm on $\mathcal{T}$ can support swap operations for logarithmic trading function in \cref{eq:swap_log} with running time big $O$ of the visiting number of $\mathcal{T}$ on $(\mathcal{X}, \mathcal{F})$.
\end{corollary}

\begin{corollary}
    Given a set system $(\mathcal{X}, \mathcal{F})$ and a partition tree $\mathcal{T}$, a lazy propagation algorithm on $\mathcal{T}$ can support swap operations for linear trading function in \cref{eq:swap_lin} with running time big $O$ of the visiting number of $\mathcal{T}$ on $(\mathcal{X}, \mathcal{F})$.
\end{corollary}

\section{Open problems and conclusion}
Based on computational geometry,  we present a unified framework for both analyzing the computational complexity and designing efficient algorithms for LMSR.   There are several directions for further exploration.  Firstly, beyond LMSR, we extend our framework to other scoring rules and show computational complexity distinctions between different scoring rules.  Further investigations into how different scoring rules impact the computational complexity of combinatorial prediction markets would be interesting.
Secondly, while our focus has been on exact arbitrage-free combinatorial prediction markets, exploring faster approximation algorithms by utilizing approximation range queries~\cite{10.1016/j.comgeo.2007.05.008} could be a promising direction.  Our multi-resolution market design offers a systematic approach to integrate multiple independent markets while upholding computational efficiency.  Exploring additional sufficient conditions for efficient and localized arbitrage removal could yield valuable insights.  
Finally, our preliminary investigation into CFMMs highlights a range of interesting variants of the query range update problem that could motivate further exploration and development.

\bibliography{main}
\bibliographystyle{plainnat}
\clearpage
\appendix

\section{Proofs in Section~\ref{sec:lmsr}}\label{app:lmsr}

\subsection{Proofs in Section~\ref{sec:alg}}

\begin{algorithm}
    \begin{algorithmic}
    \Require A range $E\subseteq \mathcal{X}$, value to update $S\in \R_+$, and a partition tree $\mathcal{T}$ on $\mathcal{X}$ where each node $v$ stores the encoding of associated node-set $N(v)\subseteq \mathcal{X}$, the list of children $\mathcal{C}(v)\subset \mathcal{V}$, weight $\val(v)$ and pending update $\pend(v)$ which are initially one, $\val(v) = \pend(v) = 1$.
        \Function{range\_update}{$E, S$}
            \State \Call{range\_update}{$E,S,root$}
        \EndFunction
        \Function{range\_update}{$E, S, v$}
        \If{$\pend(v)\neq 1$}\Comment{Check if there are pending updates for the current node}\label{line:updatecheck}
            \State $\val(v) \gets \pend(v)\cdot \val(v)$
            \For{$u\in \mathcal{C}(v)$}
                \State $\pend(u)\gets \pend(v)\cdot \pend(u)$
            \EndFor
            \State $\pend(v) \gets 1$
        \EndIf
        
        \If{$N(v)\subseteq E$}\Comment{$E$ contains $N(v)$}\label{line:update1}
            \State $\val(v) \gets S\cdot \val(v)$
            \For{$u\in \mathcal{C}(v)$}
                \State $\pend(u)\gets S\cdot \pend(u)$
            \EndFor
            \State \Return
        \ElsIf{$N(v)\cap E = \emptyset$}\Comment{$E$ and $N(v)$ are disjoint}\label{line:update2}
            \State \Return
        \Else                       \Comment{$E$ crosses the node set $N(v)$}\label{line:update3}
            \State $ans \gets 0$
            \For{$u\in \mathcal{C}(v)$}\Comment{Recursive call to all the children of $v$ node}
                \State \Call{range\_update}{$E\cap N(u), S, u$}
                \State $ans \gets ans+\val(u)$
            \EndFor
            \State $\val(v)\gets ans$
            \State \Return
        \EndIf
        \EndFunction
    \end{algorithmic}
\caption{Range update on partition trees}\label{alg:lazyUpdate}
\end{algorithm}



\begin{algorithm}\caption{Range query on partition trees}\label{alg:lazyquery}
    \begin{algorithmic}
    \Require A range $E\subseteq \mathcal{X}$, value to update $S\in \R_+$, and a partition tree $\mathcal{T}$ on $\mathcal{X}$ as \cref{alg:lazyUpdate}
        \Function{range\_query}{$E$}
            \State \Call{range\_query}{$E,root$}
        \EndFunction

        \Function{range\_query}{$E, v$}
        \If{$\pend(v)\neq 1$}\Comment{Check if there are pending updates for the current node}\label{line:querycheck}
            \State $\val(v) \gets \pend(v)\cdot \val(v)$
            \For{$u\in \mathcal{C}(v)$}
                \State $\pend(u)\gets \pend(v)\cdot \pend(u)$
            \EndFor
            \State $\pend(v) \gets 1$
        \EndIf
        \State $ans \gets 0$
        \If{$N(v)\subseteq E$}\Comment{$E$ contains $N(v)$}
            \State $ans\gets \val(v)$
        \ElsIf{$N(v)\cap E = \emptyset$}\Comment{$E$ and $N(v)$ are disjoint}
            \State $ans\gets 0$
        \Else       \Comment{$E$ crosses the node set $N(v)$}
            \For{$u\in \mathcal{C}(v)$}
                \State $ans \gets ans+$\Call{range\_query}{$E\cap N(u), u$}
            \EndFor
        \EndIf
        \State \Return $ans$
        \EndFunction
    \end{algorithmic}
\end{algorithm}
\begin{proof}[Proof of \cref{thm:visiting}]
The time complexity guarantee holds because the number of recursion of both update and range queries are exactly the visiting number of the partition tree.  We now show the correctness of the algorithm so that for any range query with $E$, the output value is correct.  

We first need to introduce some notions.  
Given a range update or range query operation with range $E$, let $U(E)\subseteq \mathcal{V}$ be the set of visited nodes.  
We further classify those nodes into three disjoint types according to lines~\ref{line:update1} to \ref{line:update3}: $U_1(E) = \{v\in U(E): N(v)\subseteq E\}$, $U_2(E) = \{v\in U(E): N(v)\cap E = \emptyset \}$ and $U_3(E) = U(E)\setminus (U_1(E)\cup U_2(E))$ be the set of first, second and third case nodes respectively.   
As the recursion stops if and only if the node is in $U_1(E)$ or $U_2(E)$, we further call first two types the boundary of those visited nodes which does not have any visited descendent nodes, and the third type non-boundary. 

Now we state three claims but defer the proofs below. 

\begin{claim}\label{lem:lazy3}
In a range update or range query operation with $E$, each visited node $v\in {U}(E)$ has unit lazy value $\pend(u) = 1$ at the end of the round.  
\end{claim}
We can easily check that that the lazy value of all visited node is reset to $1$ at the end of each round in \cref{alg:lazyUpdate,alg:lazyquery}.

\begin{claim}\label{lem:lazy2}
In a range update or range query operation with range $E$, $U_1(E)$ forms a partition of $E$
\end{claim}

\begin{claim}\label{lem:lazy1}
For all round $t$ and $v\in \mathcal{V}$, $\val(v)\prod_{u:N(v)\subseteq N(u)} \pend(u) = W^{(t)}(N(v))$ at the end of the round.  
\end{claim}

With the above three claims, the answer of \cref{alg:lazyquery} with $E$ is 
\begin{align*}
    \sum_{v\in U_1(E)}\val(v) =& \sum_{v\in U_1(E)}\frac{W^{(t)}(N(v))}{\prod_{u:N(v)\subseteq N(u)} \pend(u)}\tag{by \cref{lem:lazy1}}\\
    =& \sum_{v\in U_1(E)}{W^{(t)}(N(v))}\tag{by \cref{lem:lazy3}}\\
    =& W^{(t)}(E)\tag{by \cref{lem:lazy2}}
\end{align*}
which proves the correctness. 
\end{proof}

\begin{proof}[Proof of \cref{lem:lazy2}]
    By the definition of partition trees and $U_1(E)$, every node sets $N(v)\subset E$ and are mutually disjoint.  Hence, it is sufficient to show that for every point $x\in E$, there exists $v\in U_1(E)$ so that $x\in N(v)$.  Because $x\in N(root)\subseteq {U}(E)$, there exists a visited node $v$ that contains $x$ and has the deepest level.  However, if $v\in U_3(E)$ is a non-boundary node, in line~\ref{line:update3}, one of $v$'s children is visited and contains point $x$ which is a contradiction.  Therefore, $v\in U_1(E)$ which completes the proof.
\end{proof}

\begin{proof}[Proof of \cref{lem:lazy1}]
Given $t$, suppose the statement is correct for all round before round $t$.  Let $pend$ and $val$ be the data at the begin of round $t$ and $pend'$ and $\val'$ be the data at the end of round $t$.

If round $t$ has a range query operation with any $E$, since the correct weights is not changed, we only need to show that 
\begin{equation}\label{eq:lazy11}
    \val'(v) \prod_{u:N(v)\subseteq N(u)} \pend'(u) = \val(v) \prod_{u:N(v)\subseteq N(u)} \pend(u)
\end{equation}is also unchanged for any $v\in \mathcal{V}$.  We will use an induction on visited nodes in the DFS pre-order in \cref{alg:lazyquery} to prove \cref{eq:lazy11}.  For the base case, the root node, when $\pend(root)\neq 1$, the if statement in line~\ref{line:querycheck} 1) updates its value $\val'(root) = \val(root)\pend(root)$ and 2) propagates the lazy value to each child $u\in \mathcal{C}(root)$ so that $\pend(root)\pend(u)$ is unchanged.  The first ensures that \cref{eq:lazy11} holds for the root node and the second ensures that \cref{eq:lazy11} holds for any non root node.  The proof for the induction step follows similarly. 

On the other hand, suppose that round $t$ has a range update with range $E$ and $S\in \mathbb{R}_+$.  By the above argument, the if statement in line~\ref{line:updatecheck} does not change the value of $\val(v) \prod_{u:N(v)\subseteq N(u)} \pend(u)$ for any $v$, so we only need to consider the effects in the three cases from line~\ref{line:update1} to \ref{line:update3}.  We will use the DFS post-order in \cref{alg:lazyUpdate} where the base case consist of the boundary of visited nodes (the first two cases in line~\ref{line:update1} and \ref{line:update2}).  If a node $v\in U_1(E)$ satisfies line~\ref{line:update1}, we have
$$\val'(v) = S\cdot \val(v) = S\cdot  \val(v)\prod_{u:N(v)\subseteq N(u)} \pend(u) = S\cdot W^{(t-1)}(N(v)) = W^{(t)}(N(v))$$
where the second equality holds because of \cref{lem:lazy3}, and the third holds because the equality holds for round $t-1$.  The second case $v\in U_2(E)$ is trivial.   For the third case in line~\ref{line:update3}, $v$ will not be a leaf node and $N(u)$ for all $u\in \mathcal{C}(v)$ forms a partition of $N(v)$.  Thus, by induction hypothesis for all $u\in \mathcal{C}(v)$ we have $\val'(u) = W^{(t)}(N(u))$, and $\val'(v) = \sum_{u\in \mathcal{C}(v)} = W^{(t)}(N(v))$.  Finally, by \cref{lem:lazy3}, $\val'(v)\prod_{u:N(v)\subseteq N(u)} \pend(u) = \val'(v) = W^{(t)}(N(v))$ which completes the proof. 
\end{proof}

\begin{proof}[Proof of \cref{prop:one}]
     We only need to show the time complexity, as the correctness follows directly from \cref{thm:visiting}.  Because $E_{(i,j)}$ crosses $E_{(i',j')}$ only if $i\in E_{(i',j')}$ or $j\in E_{(i',j')}$, for all $i\le j$ and $i'\le j'$ in $\{0,\dots, n-1\}$, an interval $E_{(i,j)}$ can only cross at most two node sets in each level.  Therefore, the visiting number is at most twice of the number of level and, thus, in $O(\log n)$.
\end{proof}
\subsection{Proofs in Sections~\ref{sec:hardness}}
The partial-sums problem is to maintain an length $n$ array $W$ subject to the following operations:
\begin{enumerate}
    \item $update(k, \Delta)$: modify $W(k) \gets \Delta$.
    \item $sum(k)$: returns the partial sum $\sum_{i\le k} W(i)$.
\end{enumerate}

\begin{theorem}[Theorem 4.1 in \citet{patrascu2004tight}]\label{thm:partialsumhard}
Any algorithm for the online partial sums problem in the group arithmetic model has a
running time per operation of $\Omega(\log n)$ in the worst case. 
\end{theorem}

\begin{proof}[Proof of \cref{prop:hard1d}]
We will reduce the above partial-sum problem to the RQRU problem with interval set system.  

Let $[1:k] = \{1,\dots, k\}$ for all $k = 1,\dots, n$ and $[1:0] = \emptyset$.
\begin{itemize}
    \item For each sum query with $k$, we return $\query([1:k])$ from the RQRU algorithm.  
    \item For each update query with $k$ and $\Delta$, we compute $\delta = \frac{\Delta}{\query([1:k])-\query([1:k-1])}$ and call $\update([1:k], \delta)$ and $\update([1:k-1], 1/\delta)$ by calling the range query and update function twice from the RQRU algorithm.
\end{itemize} 
If an RQRU algorithm has $T_U(n)$ update time and $T_Q(n)$ query time, the above reduction can solve partial-sum problem in $O(T_U(n)+T_Q(n))$ times.  Therefore, $\max \{T_U(n), T_Q(n)\} = \Omega(\log n)$ by \cref{thm:partialsumhard}.
\end{proof}

\begin{proof}[Proof of \cref{prop:hard2d}]
We will reduce matrix product of $A$ and $B$ using a RQRU algorithm for two-dimensional regular orthogonal set system with $\mathcal{X} = [m]^2 = \{(i,j): i, j\in [m]\}$.  To simplify the notation, for all $i, j = 1, \dots, m$, we let $(i,j) = \{(i,j)\}$, $(:,j) := \{(k, j):k = 1, \dots, m\}$, and $(i,:) := \{(i,k):k = 1, \dots, m\}$ which are all valid two-dimensional interval ranges. 

Given $A$ and $B$ are $m$-by-$m$ matrices, we can compute $C = AB$ as the following: run $m^2$ range updates $\update((i,j), A_{i,j})$ for all $i, j = 1,\dots m$.  Let $C$ be an $m$ by $m$ matrix.  Given $j = 1,\dots, m$, we first run $\update((:, j), B_{i,j})$ for all $i$, $C_{i,j}\gets \query((i, :))$ for all 
$i$, and $\update((:,j), 1/B_{i,j})$ for all $i$.  The resulting matrix $C$ will equal $AB$.  

The initialization take $m^2$ range updates, computing each column of $C$ takes $2m$ range update and $m$ range query.  Therefore, the time complexity of matrix product can be solved in $(m^2+2m^2) T_U(m^2)+m^2 T_Q(m^2)$ which completes the proof.
\end{proof}

\begin{proof}[Proof of \cref{cor:vc}]
By \cref{prop:limit_vc}, it is sufficient to show the VC dimensions of those set systems are infinite. 
\begin{claim}\label{lem:pairing}
    Given an even number $K$ and $n = K!$, the VC dimension of the pairing set system on $K$ candidates $\{0,1,\dots, K-1\}$ is at least $\lfloor \log_2(K/2)\rfloor$ which is increasing in $K$ and $n$. \footnote{By Stirling's formula $\lfloor \log_2(K/2)\rfloor = \Omega\left(W(\frac{1}{e}\log(\frac{n}{\sqrt{2\pi}}))\right)$ where $W(\cdot)$ is the Lambert $W$ function.}
\end{claim}

\begin{proof}[Proof of \cref{lem:pairing}]
We will construct $D := \lfloor \log_2(K/2)\rfloor$ permutations $\mathcal{X}' = \{x_1, x_2,\dots, x_D\}\subset \mathcal{X}$ that is shattered by the pairing set system.  First consider a matrix $A\in \{0,1\}^{D, 2^D}$ where the columns consists of all binary strings of length $D$.  For the $l$-th permutation $x_l$, we start with the identity permutation and swap the ordering of $2j$ and $2j+1$ if $A_{l,j} = 1$ for all $j$.  Because, for all $j\le 2^D$, the set of permutations $\mathcal{X}'\cap \tau_{2j+1, 2j} = \{x_l: A_{l,j} = 1\}$ corresponds to the $j$-th column $j$ of $A$, $\mathcal{X}'$ is shattered by the pairing set system. 
\end{proof}

\begin{claim}\label{lem:junta}
    Given a positive integer $K$, the VC dimension of $1$-junta set system on $\{0,1\}^K$ is at least $\lfloor \log_2(K)\rfloor$ which is increasing in $K$.
\end{claim}

\begin{proof}[Proof of \cref{lem:junta}]
We will construct $D := \lfloor \log_2(K)\rfloor$ boolean strings $\mathcal{X}' = \{x_1, x_2,\dots, x_D\}\subset \{0,1\}^K$ that is shattered by the $1$-junta set system.  First consider a matrix $A\in \{0,1\}^{D, 2^D}$ where the columns consists of all binary strings of length $D$.  We define $x_i$ as the $i$-th row of $A$ padded up with $0$s to length $K\ge 2^D$.  Because each subset of $\mathcal{X}'$ corresponds to $j$-th  column of $A$ and can be derived by an $1$-junta function $\psi_j$, $\mathcal{X}'$ is shattered by the $1$-junta set system.
\end{proof}
    
\end{proof}
\section{Proofs in Section~\ref{sec:beyondlmsr}}
\subsection{Proofs in Section~\ref{sec:qmsr}}
\begin{proof}[Proof of \cref{lem:equivq}]
    We define a reduction from a QMSR to $(+,+)$-RQRU in \cref{def:rqruq} with $l = 2$.  Given an initial state $\vw^{(0)}$ on $(\mathcal{X}, \mathcal{F})$, we run the algorithm on initial weight $Z^{(0)}$ where $Z^{(0)}(x) = \begin{bmatrix}1\\
        w^{(0)}_x
    \end{bmatrix}$ for all $x\in \mathcal{X}$ and compute $M = \sum_{x\in \mathcal{X}} w^{(0)}_x$.  
    \begin{itemize}
        \item For each price operation with $E\in \mathcal{F}$, we run $\query_Q(E) = \begin{bmatrix}
            \Sigma_0\\ \Sigma_1
        \end{bmatrix}$ by calling the range query function one from the RQRU algorithm and return $$\frac{1}{n}\Sigma_0+\frac{1}{2b}\Sigma_1-\frac{1}{2bn}\Sigma_0M.$$  
        \item For each buy operation with $E\in \mathcal{F}$ and share $s\in \mathbb{R}$, we run the range update $\update_Q(E, \begin{bmatrix}
            0\\ s
        \end{bmatrix})$, from the RQRU algorithm, and update $M\gets M+|E|s$
        \item Finally, to compute a cost operation with set $E$ and share $s$, we run $\query_Q(E) = \begin{bmatrix}
            \Sigma_0\\ \Sigma_1
        \end{bmatrix}$ and return $$\left(\frac{s}{n}+\frac{s^2}{4b}\right)\Sigma_0-\frac{s}{4bn}\Sigma_0^2+\frac{s}{2b}\Sigma_1-\frac{s}{2bn}\Sigma_0M.$$ 
    \end{itemize}
    Because the buy operation also needs to update $M$ that takes additional range query to compute $|E|$, the time complexity for buy operation is $O(T_Q(n)+T_U(n))$.  The complexity for price and cost are straightforward.  
    
    To prove the correctness, we first use induction on the sequence of operations to show the following invariant: for all round $t$,
    \begin{equation}\label{eq:invq}
        Z^{(t)}(x) = \begin{bmatrix}
            1\\ w^{(t)}_x
        \end{bmatrix}\text{ for all }x\in \mathcal{X}\text{ and } M^{(t)} = \sum_{x\in \mathcal{X}} w^{(t)}_x
    \end{equation}
    The based case holds by initialization.  If we encounter a buy operation with $E$ and $s$ at round $t+1$, the share of $x\in E$ is updated from $w^{(t)}_x$ to $w^{(t+1)} = w^{(t)}_x+s$, and the above reduction also updates $Z^{(t)}(x)$ to $Z^{(t+1)}(x) = Z^{(t)}(x) +\begin{bmatrix}
        0\\s
    \end{bmatrix} = \begin{bmatrix}
        1\\w^{(t+1)}_x
    \end{bmatrix}$.  The equality also holds for all $x\notin E$.  Finally, 
    $$M^{(t+1)} = M^{(t)}+|E|s = \sum_{x\in E} (w^{(t)}(x)+s)+\sum_{x\in \mathcal{X}\setminus E} w^{(t)}(x) = \sum_{x\in \mathcal{X}} w^{(t+1)}_x.$$  Thus, we prove \cref{eq:invq}.  
    
    We then show the reduction answers price and cost queries correctly.  Given a price operation with $E$ at round $t$, the reduction returns 
    \begin{align*}
        \frac{1}{n}\Sigma_0+\frac{1}{2b}\Sigma_1-\frac{1}{2bn}\Sigma_0M
        = \frac{1}{n} |E|+\frac{1}{2b}\sum_{x\in E}w^{(t)}_x-\frac{1}{2bn}|E|\sum_{x\in \mathcal{X}}w^{(t)}_x\tag{by \cref{eq:invq}}
    \end{align*}
    which equals $\price_Q(E;\vw^{(t)})$ in \cref{def:qmsr}.  For the cost operation with $E$ and $s$, note that 
    \begin{align*}
        \cost_Q(E, s; \vw^{(t)}) :=& C_Q(\vw^{(t)}+s\mathbf{1}_E)-C_Q(\vw^{(t)})\\
        =& \frac{s}{n}|E|+\frac{1}{4b}\sum_{x\in E}s(2w^{(t)}_x+s)-\frac{1}{4bn}s|E|(s|E|+2\sum_{x\in \mathcal{X}}w^{(t)}_x).\\
        =&\left(\frac{s}{n}+\frac{s^2}{4b}\right)|E|-\frac{s}{4bn}|E|^2+\frac{s}{2b}\sum w^{(t)}_x-\frac{s}{2bn}|E|\sum w^{(t)}_x\\
        =&\left(\frac{s}{n}+\frac{s^2}{4b}\right)\Sigma_0-\frac{s}{4bn}\Sigma_0^2+\frac{s}{2b}\Sigma_1-\frac{s}{2bn}\Sigma_0M\tag{by \cref{eq:invq}}
    \end{align*}
    which is identical to the output of the reduction.
\end{proof}

\subsection{Proofs in Section~\ref{sec:pmsr}}
\begin{proof}[Proof of \cref{prop:pmsr_closedform}]
    By KKT conditions, the optimal $p\in \Delta_\mathcal{X}$ satisfies $\sqrt{p(x)} = \frac{2}{3b}(w(x)-\lambda)\ge 0$ for some $\lambda$ so that $\sum_{x\in \mathcal{X}} p(x) = 1$.  Let $M_1 = \sum_x w(x)$ and $M_2 = \sum_x w(x)^2$, $M_3 = \sum_x w(x)^3$, and $\mu = \sqrt{M_1^2-n(M_2-\frac{9b^2}{4})}$. By direct calculation, $\lambda = \frac{1}{n}(M_1-\mu)$.  Therefore, 
    \begin{align*}
        C_{\frac32}(w) =& \max_{p\in \Delta_\mathcal{X}} \sum_{x\in \mathcal{X}} w(x)p(x)-b\sum_{x\in \mathcal{X}}p(x)^{\frac32}\\
        =& \sum_x \frac{4}{9b^2}(w(x)-\lambda)^2w(x)-\sum_x \frac{8}{27b^2}(w(x)-\lambda)^3\\
        =&\frac{4}{27b^2} \sum_x 3(w(x)-\lambda)^2w(x)-2(w(x)-\lambda)^3\\
        =&\frac{4}{27b^2} \sum_x w(x)^3-3\lambda^2w(x)+2\lambda^3\\
        =&\frac{4}{27b^2} \left(M_3-\frac{1}{n^2}(M_1^3-3M_1\mu^2+2\mu^3)\right)
    \end{align*}
    For the price function,
    \begin{align*}
        &\frac{\partial}{\partial w(x)}C_{\frac32}(w)\\
        =&\frac{4}{27b^2}\left(3w(x)^2-\frac{1}{n^2}(3M_1^2-3\mu^2-6M_1\mu\frac{\partial\mu}{\partial w(x)}+6\mu^2\frac{\partial\mu}{\partial w(x)})\right)\\
        =&\frac{4}{9b^2}w(x)^2-\frac{4}{27b^2n^2}\left(3M_1^2-3\mu^2-6M_1\mu\frac{M_1-nw(x)}{\mu}+6\mu^2\frac{M_1-nw(x)}{\mu}\right)\tag{$\frac{\partial\mu}{\partial w(x)} = \frac{M_1-nw(x)}{\mu}$}\\
        =&\frac{4}{9b^2}w(x)^2-\frac{4}{9b^2n^2}\left(M_1^2-\mu^2-2M_1^2+2M_1nw(x)+2M_1\mu -2\mu nw(x)\right)\\
        =&\frac{4}{9b^2}w(x)^2-\frac{8}{9b^2n}\left(M_1 -\mu \right)w(x)-\frac{4}{9b^2n^2}\left(M_1^2-2M_1^2+2M_1\mu-M_1^2+nM_2-\frac{9nb^2}{4}\right)\\
        =&\frac{1}{n}+\frac{4}{9b^2}\left(w(x)^2-\frac{2}{n}\left(M_1 -\mu \right)w(x)-\frac{1}{n^2}\left(2M_1\mu-2M_1^2+nM_2\right)\right)
    \end{align*}
    which completes the proof.
\end{proof}

\begin{proof}[Proof of \cref{lem:equivp}]
    We define a reduction from a $\frac32$-power MSR market to the RQRU problem in \cref{def:rqrup}.  Given an initial state $\vw^{(0)}$ on $(\mathcal{X}, \mathcal{F})$, we run the algorithm on initial weight $Z^{(0)}$ where $Z^{(0)}(x) = \begin{bmatrix}1\\
        w^{(0)}_x\\\left(w^{(0)}_x\right)^2\\\left(w^{(0)}_x\right)^3
    \end{bmatrix}$ for all $x\in \mathcal{X}$, and compute $M = \begin{bmatrix}
        M_0\\M_1\\M_2\\M_3
    \end{bmatrix}$ so that $M_0 = \sum_{x\in \mathcal{X}}1$, $M_1 = 
        \sum_{x\in \mathcal{X}} w^{(0)}_x$, $M_2 = \sum_{x\in \mathcal{X}} \left(w^{(0)}_x\right)^2$, $M_3 = \sum_{x\in \mathcal{X}} \left(w^{(0)}_x\right)^3$, and $\mu = \sqrt{M_1^2-n(M_2-\frac{9}{4})}$.  
    \begin{itemize}
        \item For each price operation with $E\in \mathcal{F}$, we run $\query_Q(E) = \begin{bmatrix}
            \Sigma_0\\ \Sigma_1\\ \Sigma_2\\ \Sigma_3
        \end{bmatrix}$ by calling the range query function one from the RQRU algorithm and return $$\frac{1}{n}\Sigma_0+\frac{4}{9}\left(\Sigma_2+\frac{2}{n}(\mu-M_1) \Sigma_1-\frac{1}{n}\Sigma_2+\frac{2}{n^2}M_1^2\Sigma_0-\frac{2}{n^2}M_1\mu\Sigma_0\right).$$  
        \item For each buy operation with $E\in \mathcal{F}$ and share $s\in \mathbb{R}$, we run the range update $\update_{\frac32}(E, s)$, from the RQRU algorithm, and update $M\gets \alpha_s(M)$ where $\alpha_s$ is defined in \cref{def:rqrup}.
        \item Finally, to compute a cost operation with set $E$ and share $s$, we run the following three steps: First, run $\update(E,s)$ and compute $c' = \frac{4}{27}\left( M_3-\frac{1}{n^2}(M_1^3-3M_1^2\mu+2\mu^3  )\right)$.  Second, run $\update(E, -s)$ and compute $c = \frac{4}{27}\left( M_3-\frac{1}{n^2}(M_1^3-3M_1^2\mu+2\mu^3  )\right)$. Third, return $c'-c$. 
    \end{itemize}
    The time complexity is straightforward.  
    
    To prove the correctness, we first use induction on the sequence of operations to show the following invariant: for all round $t$,
    \begin{equation}\label{eq:invp}
        Z^{(t)}(x) = \begin{bmatrix}
            Z^{(t)}_{0}\\Z^{(t)}_{1}\\Z^{(t)}_{2}\\Z^{(t)}_{3}
        \end{bmatrix} = \begin{bmatrix}
            1\\ w^{(t)}_x\\\left(w^{(t)}_x\right)^2\\ \left(w^{(t)}_x\right)^3
        \end{bmatrix}\text{ for all }x\in \mathcal{X}\text{ and } M^{(t)} = \begin{bmatrix}
            M^{(t)}_{0}\\M^{(t)}_{1}\\M^{(t)}_{2}\\M^{(t)}_{3}
        \end{bmatrix} = \begin{bmatrix}
            \sum_{x\in \mathcal{X}}1\\ \sum_{x\in \mathcal{X}}w^{(t)}_x\\\sum_{x\in \mathcal{X}}\left(w^{(t)}_x\right)^2\\\sum_{x\in \mathcal{X}}\left(w^{(t)}_x\right)^3
        \end{bmatrix}
    \end{equation}
    The based case holds by initialization.  If we encounter a buy operation with $E$ and $s$ at round $t+1$, the share of $x\in E$ is updated from $w^{(t)}_x$ to $w^{(t+1)} = w^{(t)}_x+s$, and the above reduction also updates $Z^{(t)}(x)$ to 
    
    \begin{align*}
        Z^{(t+1)}(x) =& \alpha_s(Z^{(t)}(x)) = \begin{bmatrix} Z^{(t)}_{0}\\ Z^{(t)}_{1}+s\\Z^{(t)}_{2}+2sZ^{(t)}_{1}+s^2\\Z^{(t)}_{3}+3sZ^{(t)}_{2}+3s^2Z^{(t)}_{1}+s^3\end{bmatrix}.\\
        =& \begin{bmatrix} 1\\ w^{(t)}_x+s\\\left(w^{(t)}_x\right)^2+2sw^{(t)}_x+s^2\\\left(w^{(t)}_x\right)^3+3s\left(w^{(t)}_x\right)^2+3s^2w^{(t)}_x+s^3\end{bmatrix}\\
        =& \begin{bmatrix} 1\\ w^{(t)}_x+s\\(w^{(t)}_x+s)^2\\(w^{(t)}_x+s)^3\end{bmatrix}=\begin{bmatrix}
            1\\ w^{(t+1)}_x\\ \left(w^{(t+1)}_x\right)^2\\\left(w^{(t+1)}_x\right)^3
        \end{bmatrix}
    \end{align*}  The equality also holds for all $x\notin E$.  Finally, 
    $M^{(t+1)} = \alpha_s(M^{(t)}).$ by the same computation.  Thus, we prove \cref{eq:invq}.  
    The reduction answers price and cost queries correctly follows directly from \cref{prop:pmsr_closedform} and \cref{eq:invp}.
\end{proof}

\begin{proof}[Proof of \cref{thm:pmsr}]
    Because $(+,\alpha)$-RQRU satisfy \cref{def:rqrug}, combining \cref{thm:rqrug} and \cref{lem:equivp} completes the proof.
\end{proof}

\subsection{Proofs in Section~\ref{sec:multi-res}}

\begin{proof}[Proof of \cref{prop:multi-resolution-cost}]
We first prove that the constraints $\A^\top \mu = \vzero$ imply that all submarkets $k=0,1,\dotsc,K$ are mutually coherent.
To do this, it suffices to show that all pairs of submarkets at consecutive levels $\ell < K$ and $\ell+1$ are coherent, i.e., $\mu_u=\sum_{v\in\mathcal{C}(u)} \mu_v$ for all $u\in\V_\ell$ as we define $\A$ in \cref{eq:A_gen}.
As prices at level $K$ are determined by $C_K$, they describe a probability distribution over $\mathcal{X}$.
Since all submarkets are coherent with the finest submarket $K$, $\mu$ is a coherent price vector for $\mathcal{N} = \cup_{k = 0}^K \mathcal{N}_k$.

Consider a fixed $w$ and the corresponding $\eta^*$ that minimizes \cref{eq:lcmm_cost}.
We calculate prices over $\mathcal{N} = \cup_{k = 0}^K \mathcal{N}_k$ as
$
\vp(w)=\nabla C(w)=\nabla\Cx(w+\A\eta^*).
$
By the first order optimality, $\eta^*$ minimizes \cref{eq:lcmm_cost} if and only if $\A^\top\bigParens{\nabla\Cx(w+\A\eta^*)}=\vzero$.
This means that $\A^\top\vp(w)=\vzero$, and thus arbitrage opportunities expressed by $\A$ are completely removed by the cost function $C$ in \cref{eq:lcmm_cost}.

\end{proof}

\begin{proof}[Proofs for     \cref{ex:multiresolution_lmsr_arb}]
To calculate price, we consider a node $v \neq root$ with $k=\level(v)$. 
We have
\begin{equation}
\label{eq:pxz:2}
\price(N(v)) 
=\frac{e^{\tilde{w}(v)/b_k}}{e^{\tilde{w}(v)/b_k}+\sum_{u \in \sib(v)} e^{\tilde{w}(u)/b_k}} \cdot \price(N(\pt(v))).
\end{equation}
Following the construction of $\A^{\text{LMSR}}$ in \eqref{eq:A} and expanding $\tilde{\vw}$, we get
\begin{equation}
\label{eq:ttheta}
\tilde{w}(v) = w(v) + \sum_{u\in\U} A_{vu}\eta(u)
= w(v) + B_{k}\eta(v) - b_{k}\sum_{u \supset v} \eta(u).
\end{equation}
and
\begin{equation}
\label{eq:pxz:3}
\price(N(v)) 
=\frac{\exp\BigParens{\frac{w_v+B_{k}\eta_v}{b_k}}}
{\exp\BigParens{\frac{w_v+B_{k}\eta_v}{b_k}}
	+
\sum_{u \in \sib(v)} \exp\BigParens{\frac{w_u+B_{k}\eta_u}{b_k}}} \cdot \price(N(\pt(v))).
\end{equation}

Next, we show that given a price coherent market, after a trader buys $s$ shares of security associated with node $u$, the price incoherence between the submarket at $\ell := \level(u)$ and submarkets at all other levels can be removed efficiently.
Specifically, to restore price coherence, it suffices to update $\eta(u)$ by a closed-form amount.
%
Consider two arbitrary levels $k$ and $m$ with $\ell < k < m \leq K$.
Since prices are coherent between levels $k$ and $m$ before buying $x$ shares of bundle $a_u$, we have, for any $v\in\V_k$,
\[ p_v = \sum_{z \in \V_m:\:z\subset v} p_z.\]
Let $\tilde{w}^*=\tilde{w}+x a_u$.
Based on matrix $\A^{\text{LMSR}}$, we have
\[
	\tilde{w}^*(v)=
	\begin{cases}
	\tilde{w}(v)-x b_{k}
	&\text{if $v \subset u$,}
	\\
	\tilde{w}(v)
	&\text{otherwise,}
	\end{cases}
	\qquad\qquad
	\tilde{w}^*(z)=
	\begin{cases}
	\tilde{w}(z)-x b_{m}
	&\text{if $z \subset u$,}
	\\
	\tilde{w}(z)
	&\text{otherwise.}
	\end{cases}
\]\fang{what is the difference?}
We calculate the new price $p^*_v$ of any node $v\in\V_{k}$ and show it equals to the price derived from its descendants $z \in \V_{m}$. 
First, if $v \subset u$,
\begin{align*}
	p^*_v &=
	\frac{p_v e^{-x}}{p_u e^{-x}+1-p_u}
	= \frac{\sum_{z\in\V_{m}:\:z \subset v} p_z e^{-x}}{p_u e^{-x}+1-p_u}
	= \sum_{z\in\V_{m}:\:z \subset v} p^*_z.
	\intertext{%
		If $v \not\subset u$, then we similarly have}
	p^*_v &=
	\frac{p_v}{p_u e^{-x}+1-p_u}
	= \frac{\sum_{z\in\V_{m}:\:z \subset v} p_z}{p_u e^{-x}+1-p_u}
	= \sum_{z\in\V_{m}:\:z \subset v} p^*_z.
\end{align*}
Thus, prices remain coherent among all levels $m>k>\ell$.

Next, it remains to show that prices are coherent among levels $\ell$ and $\ell+1$, i.e., $p^*_\mu = \sum_{v \in C(\mu)} p^*_v$.
Based on matrix $\A$, we have
\[
	\tilde{w}^*(\mu)=
	\begin{cases}
	\tilde{w}(\mu)+xB_{\ell}
	&\text{if $\mu = u$,}
	\\
	\tilde{w}(\mu)
	&\text{otherwise,}
	\end{cases}
	\qquad\qquad
	\tilde{w}^*(v)=
	\begin{cases}
	\tilde{w}(v)-x b_{\ell+1}
	&\text{if $z \subset u$,}
	\\
	\tilde{w}(v)
	&\text{otherwise.}
	\end{cases}
\]\fang{can it be three cases?}
To verify \cref{eq:arb_x}, we have
\begin{align}
    p^*_\mu &= \sum_{v \in \mathcal{C}(\mu)} p^*_v\\
    \frac{p_\mu e^{xB_\ell/b_\ell}}{p_\mu e^{xB_\ell/b_\ell}+1-p_\mu} &= \frac{\sum_{v \in C(\mu)}p_v e^{-x}}{\sum_{v \in C(\mu)}p_v e^{-x}+1-\sum_{v \in C(\mu)}p_v}\\
    x &= \frac{b_\ell}{B_{\ell-1}}\ln\Parens{\frac{1-p_\mu}{p_\mu}\cdot\frac{\sum_{v \in C(\mu)}p_v}{1-\sum_{v \in C(\mu)}p_v}}
\end{align}
Note that $B_{\ell-1} = B_{\ell} + b_{\ell}$.
\end{proof}
\begin{proof}[Proof for \cref{lem:lcmm_q}]
    For any internal node $u\in \mathcal{V}$ in the partition tree associated with the filtration $(\mathcal{N}_k)_{k = 0,\dots K}$\fang{define?}, we use the bundle defined in \cref{eq:A}
    $a_u^*\in \R^{|\mathcal{V}|}$ where 
    $$a_{u,v}^* = \begin{cases}
        B_{\level(u)} = \sum_{k>\level(u)}^K b_{k}& \text{if $v=u$,}
\\
-b_{\level(v)}&\text{if $v$ is a descendent of $u$,}
\\
0 &\text{otherwise.}
    \end{cases}$$
    Similar to \cref{ex:multiresolution_lmsr}, these bundles are in the column space of $\mA$.\fang{better words?}

    Now we show how to make prices coherent locally by the following claim. As we can scale $\vb$ linearly, we can assume the original $C_q$ has liquidity parameter $1$.  
    Given $\xi\in \R$, $u$, $u'\neq u$ that is not ascendant of $u$, and $\tilde{w}^0$ with associated prices $p^0 = \tilde{p}(\tilde{w}^0)$ and $p = \tilde{p}(\tilde{w}^0+\xi a_u^*)$, 
    $$p_{u'}- \sum_{v'\in \mathcal{C}(u')}p_{v'} = p^0_{u'}- \sum_{v'\in \mathcal{C}(u')}p^0_{v'}.$$
    Moreover, if $\xi = \xi_u^* = \frac{b_\ell}{\sum_{k = \ell}^K b_k}\frac{2n(\sum_{v\in \mathcal{C}(u)} p_v^0-p_u^0)}{N(u)(n-N(u))}$
    $$p_{u} = \sum_{v\in \mathcal{C}(u)}p_{v}.$$  In other words, we can make prices coherent between $u$ and its children without affecting price coherence for any non-ascendant $u'$.  Therefore, with the above claim, we can iteratively remove arbitrage by using the above bundle for each node in the DFS order. 

    For the first part, if $u'$ is a descendant of $u$ with $\level(u') = k$ by \cref{eq:QMSR} and the definition of $a_u^*$,
    \begin{align*}
        p_{u'} =& \frac{1}{n}|N(u')|+\frac{1}{2b_k}\sum_{x\in N(u')} (\tilde{w}^0_k(x)-\xi b_k)-\frac{|N(u')|}{2b_kn}\left(\sum_{x\in N(u)} (\tilde{w}^0_k(x)-\xi b_k)+\sum_{x\notin N(u)} \tilde{w}^0_k(x)\right)\\
        =& p_{u'}^0-\frac{\xi(n-|N(u)|)}{2n} |N(u')|
    \end{align*}
    Similarly for any $v'\in \mathcal{C}(u')$, $p_{v'} = p_{v'}^0-\frac{\xi(n-|N(u)|)}{2n}|N(v')|$, and 
    \begin{align*}
        p_{u'}-\sum_{v'\in \mathcal{C}(u')} p_{v'} =& p_{u'}^0-\frac{\xi(n-|N(u)|)}{2n} |N(u')|-\sum_{v'\in \mathcal{C}(u')} p_{v'}^0+\frac{\xi(n-|N(u)|)}{2n}\sum_{v'\in \mathcal{C}(u')}|N(v')|\\
        =& p_{u'}^0-\sum_{v'\in \mathcal{C}(u')} p_{v'}^0
    \end{align*}
    If $u'$ with $\level(u') = k$ is neither descendant nor ascendant of $u$, $N(u')$ is disjoint to $N(u)$, we have 
    \begin{align*}
        p_{u'} =& \frac{1}{n}|N(u')|+\frac{1}{2b_k}\sum_{x\in N(u')} \tilde{w}^0_k(x)-\frac{|N(u')|}{2b_kn}\left(\sum_{x\in N(u)} (\tilde{w}^0_k(x)-\xi b_k)+\sum_{x\notin N(u)} \tilde{w}^0_k(x)\right)\\
        =& p_{u'}^0+\frac{\xi(|N(u)|)}{2n} |N(u')|
    \end{align*}
    and the rest follows the identical argument. 

    Finally, for $u$ with $\level(u) = \ell$ we have
    \begin{align*}
        p_{u} =& \frac{1}{n}|N(u)|+\frac{1}{2b_\ell}\sum_{x\in N(u')} (\tilde{w}^0_\ell (x)+\xi B_\ell)-\frac{|N(u)|}{2b_\ell n}\left(\sum_{x\in N(u)} (\tilde{w}^0_\ell(x)+\xi B_\ell)+\sum_{x\notin N(u)} \tilde{w}^0_\ell(x)\right)\\
        =&p_u^0+\frac{ \xi(n-|N(u)|)|N(u)|B_\ell}{2n b_\ell}, 
    \end{align*} and $p_{v} = p_{v}^0-\frac{\xi(n-|N(u)|)}{2n}|N(v)|$ for all $v\in \mathcal{C}(u)$.  Therefore by taking $\xi = \xi_u^*$
    \begin{align*}
        p_u-\sum_{v\in \mathcal{C}(u)} p_v =& p_u^0-\sum_{v\in \mathcal{C}(u)} p_v^0+\frac{ \xi(n-|N(u)|)|N(u)|B_\ell}{2n b_\ell}+\frac{\xi_u^*(n-|N(u)|)|N(u)|}{2n}\\
        =& p_u^0-\sum_{v\in \mathcal{C}(u)} p_v^0+\xi_u^* \frac{(n-|N(u)|)|N(u)|}{2n}\left(\frac{B_\ell}{b_\ell}+1\right)\\
        =& 0
    \end{align*}
    which completes the proof.
\end{proof}
\begin{proof}[Proof of \cref{thm:multi-res}]
    With efficient and local arbitrage removal (\cref{def:local_update}), we can localize the arbitrage weight updates to the subtree rooted at the node $u$, when securities associated with $N(u)$ is traded. 
    We can continue using the local update property to go up, back along the search path of node $u$ to retain price coherence of the multi-resolution market. 
    In an arbitrage-free multi-resolution market, with each node $u$ storing its trade weight $w(u)$ and arbitrage weight $\eta(u)$, price can be computed recursively along the search path.
    Then by \cref{thm:visiting}, price, buy, and cost operation can be supported in time big $O$ of the visiting number of the constructed partition tree $\mathcal{T}$ for the multi-resolution market.
\end{proof}

\section{Generalized RQRU problem and partition tree scheme}\label{app:rqrug}
We first define the range query range update problem with general query and update operations.  

\begin{Definition}[Generalized RQRU]\label{def:rqrug}
    Let $(\mathcal{Z}, \oplus)$ be a commutative group with zero $0_\mathcal{Z}$ and a group $(\mathcal{S}, \circ)$ with identity $1_\mathcal{S}$ acts on $\mathcal{Z}$ denoted as $S\otimes Z$.  Moreover, the group action $\otimes$ and $\oplus$ satisfy the distributive law where $S\otimes (Z\oplus Z') = (S\otimes Z)\oplus(S\otimes Z')$ for all $S\in \mathcal{S}$ and $Z, Z'\in \mathcal{Z}$.  The range query and range update problem on $(\mathcal{X}, \mathcal{F})$ with $(\mathcal{Z}, \oplus)$ and $(\mathcal{S}, \circ)$ requests a sequence of operations, taking one of the following forms: for any $E\in \mathcal{F}$ and $S\in \mathcal{S}$
\begin{itemize}
    \item $\query_G(E; Z)$: return the total weight of range $E$, $Z(E) = \oplus_{x \in E} Z(x)$ where $Z(\emptyset) := 0_\mathcal{Z}$.
    \item $\update_G(E, S; W)$: for each $x \in E$, update $Z(x) \gets S\otimes Z(x)$, and for each $x'\notin E$, $Z(x') \gets Z(x')$.  
\end{itemize}
\end{Definition}

Note that the above definition generalizes all the RQRU problems in \cref{def:rqru,def:rqruq,def:rqrup}.
Interestingly, the partition tree scheme in \cref{alg:lazyquery,alg:lazyUpdate} directly adapts to this generalized RQRU problem without incurring any additional overhead defined in \cref{alg:lazygquery,alg:lazygUpdate}. 
\begin{theorem}\label{thm:rqrug}
    Given a set system $(\mathcal{X}, \mathcal{F})$ and a partition tree $\mathcal{T}$, the query time $T_Q(n)$ of \cref{alg:lazygquery} and the update $T_U(n)$ of \cref{alg:lazygUpdate} on $\mathcal{T}$ are big $O$ of the visiting number of $\mathcal{T}$ on $(\mathcal{X}, \mathcal{F})$. 


\end{theorem}

\begin{lemma}
    RQRU problems in \cref{def:rqru,def:rqruq,def:rqrup} are generalized RQRU.
\end{lemma}
\begin{proof}
    For \cref{def:rqru}, we take $\mathcal{Z} = \mathcal{S} = \R_{\ge 0}$ non-negative real numbers, $\oplus = +$ addition and $\circ = \otimes = \cdot$ as multiplication. 
\end{proof}

\begin{algorithm}\caption{Range update on partition trees}\label{alg:lazygUpdate}
    \begin{algorithmic}[1]
    \Require A partition tree $\mathcal{T}$ on $\mathcal{X}$, range $E\subseteq \mathcal{X}$, and value to update $S\in \R_+$
        \Function{range\_update}{$E, S$}
            \State \Call{range\_update}{$E,S,root$}
        \EndFunction
        \Function{range\_update}{$E, S, v$}
        \If{$\pend(v)\neq 1_\mathcal{S}$}\Comment{Check if there are pending updates for the current node}\label{line:gupdatecheck}
            \State $val(v) \gets {\pend(v)}\otimes val(v)$
            \For{$u\in \mathcal{C}(v)$}
                \State $\pend(u)\gets \pend(v)\circ \pend(u)$
            \EndFor
            \State $\pend(v) \gets 1_\mathcal{S}$
        \EndIf
        
        \If{$N(v)\subseteq E$}\Comment{$E$ contains $N(v)$}\label{line:gupdate1}
            \State $val(v) \gets S\otimes val(v)$
            \For{$u\in \mathcal{C}(v)$}
                \State $\pend(u)\gets S\circ \pend(u)$
            \EndFor
            \State \Return
        \ElsIf{$N(v)\cap E = \emptyset$}\Comment{$E$ and $N(v)$ are disjoint}\label{line:gupdate2}
            \State \Return
        \Else                       \Comment{$E$ crosses the node set $N(v)$}\label{line:gupdate3}
            \State $ans \gets 0_\mathcal{Z}$
            \For{$u\in \mathcal{C}(v)$}\Comment{Recursive call to all the children of $v$ node}
                \State \Call{range\_update}{$E\cap N(u), S, u$}
                \State $ans \gets ans\oplus val(u)$
            \EndFor
            \State $val(v)\gets ans$
            \State \Return
        \EndIf
        \EndFunction
    \end{algorithmic}
\end{algorithm}



\begin{algorithm}\caption{Range query on partition trees}\label{alg:lazygquery}
    \begin{algorithmic}[1]
    \Require A partition tree $\mathcal{T}$ on $\mathcal{X}$, and range $E\subseteq \mathcal{X}$
        \Function{range\_query}{$E$}
            \State \Call{range\_query}{$E,root$}
        \EndFunction

        \Function{range\_query}{$E, v$}
        \If{$\pend(v)\neq 1_\mathcal{S}$}\Comment{Check if there are pending updates for the current node}\label{line:gquerycheck}
            \State $val(v) \gets {\pend(v)}\otimes val(v)$
            \For{$u\in \mathcal{C}(v)$}
                \State $\pend(u)\gets \pend(v)\circ \pend(u)$
            \EndFor
            \State $\pend(v) \gets 1_\mathcal{S}$
        \EndIf
        \State $ans \gets 0_\mathcal{Z}$
        \If{$N(v)\subseteq E$}\Comment{$E$ contains $N(v)$}
            \State $ans\gets val(v)$
        \ElsIf{$N(v)\cap E = \emptyset$}\Comment{$E$ and $N(v)$ are disjoint}
            \State $ans\gets 0_\mathcal{Z}$
        \Else       \Comment{$E$ crosses the node set $N(v)$}
            \For{$u\in \mathcal{C}(v)$}
                \State $ans \gets ans\oplus $ \Call{range\_query}{$E\cap N(u), u$}
            \EndFor
        \EndIf
        \State \Return $ans$
        \EndFunction
    \end{algorithmic}
\end{algorithm}

\begin{proof}[Proof of \cref{thm:rqrug}]
    
Now we prove \cref{thm:rqrug}.  The time complexity guarantee is hold by the above discussion and \cref{thm:visiting} as the number of recursion of both update and range queries are exactly the visiting number of the partition tree.  We only need to show the correctness of the algorithm so that for any range query with range $E$ the output value is correct.  

For correctness, we use similar notions as the proof of \cref{thm:visiting}.  Given a range update or range query operation with range $E$, let $U(E)\subseteq \mathcal{V}$ be the set of visited nodes, $U_1(E) = \{v\in U(E): N(v)\subseteq E\}$, $U_2(E) = \{v\in U(E): N(v)\cap E = \emptyset \}$ and $U_3(E) = U(E)\setminus (U_1(E)\cup U_2(E))$ be the set of first, second and third case nodes respectively.

\begin{claim}\label{lem:lazyg3}
In a range update or range query operation with $E$, each visited node $v\in {U}(E)$ has unit lazy value $\pend(u) = 1_\mathcal{S}$ at the end of the round where $1_\mathcal{S}$ is the identity of group $(\mathcal{S}, \circ)$.  
\end{claim}
We can easily check that that the lazy value of all visited node is reset to $1_\mathcal{S}$ at the end of each round in \cref{alg:lazygUpdate,alg:lazygquery}.

\begin{claim}\label{lem:lazyg2}
In a range update or range query operation with range $E$, $U_1(E)$ forms a partition of $E$
\end{claim}

\begin{claim}\label{lem:lazyg1}
For all round $t$ and $v\in \mathcal{V}$ with the path from the root to itself $u_0 = root, \dots, u_k = v$, we have $\left(\prod_{i= 0}^k \pend(u_i)\right)\otimes val(v) = W^{(t)}(N(v))$ at the end of the round where $\prod$ is the product under $\circ$ in group $(\mathcal{S}, \circ)$.  
\end{claim}
\begin{proof}[Proof of \cref{lem:lazyg1}]
Given $t$, suppose the statement is correct for all round before round $t$.  Let $\pend$ and $\val$ be the data at the begin of round $t$ and $\pend'$ and $\val'$ be the data at the end of round $t$.

If round $t$ has a range query operation with any $E$, since the correct weights is not changed, we only need to show that 
\begin{equation}\label{eq:lazyg11}
   \left(\prod_{i= 0}^k \pend'(u_i)\right)\otimes \val'(v)  = \left( \prod_{i= 0}^k \pend(u_i)\right)\otimes \val(v)
\end{equation}is also unchanged for any $v\in \mathcal{V}$  with the path from the root to itself $u_0 = root, \dots, u_k = v$.  We will use an induction on visited nodes in the DFS pre-order in \cref{alg:lazygquery} to prove \cref{eq:lazyg11}.  For the base case, the root node, when $\pend(root)\neq 1_\mathcal{S}$, the if statement in line~\ref{line:gquerycheck} 1) updates its value $\val'(root) = \pend(root)\otimes \val(root)$, 2) propagates the lazy value to each child $u\in \mathcal{C}(root)$ so that $\pend(root)\circ \pend(u)$ is unchanged, and 3) $\pend'(root) = 1_\mathcal{S}$.  The first and the third ensure that \cref{eq:lazyg11} holds for the root node because $$\pend'(u)\otimes \val'(root) = 1_\mathcal{S}\otimes \val'(root) = \val'(root) = \pend(root)\otimes \val(root)$$
and the second ensures that \cref{eq:lazyg11} holds for any non root node with the path $u_0,\dots, u_k$ as
\begin{align*}
    &\left(\prod_{i= 0}^k \pend'(u_i)\right)\otimes \val'(v)\\
    =& \left(\pend'(root)\circ \pend'(u_1)\circ \prod_{i= 2}^k \pend'(u_i)\right)\otimes \val'(v)\tag{associative law of $\circ$}\\
    =& \left(\pend'(root)\circ \pend'(u_1)\circ \prod_{i= 2}^k \pend(u_i)\right)\otimes \val(v)\tag{update the root and its children's pend}\\
    =& \left(\pend(root)\circ \pend(u_1)\circ \prod_{i= 2}^k \pend(u_i)\right)\otimes \val(v)\tag{by the second property}\\
    =& \left(\prod_{i= 0}^k \pend(u_i)\right)\otimes \val(v)
\end{align*}  The proof for the induction step follows similarly. 

On the other hand, suppose that round $t$ has a range update with range $E$ and $S\in \mathcal{S}$.  By the above argument, the if statement in line~\ref{line:gupdatecheck} does not change the value of $\left(\prod_{i = 0}^k \pend(u_i)\right)\otimes \val(v)$ for any $v$, so we only need to consider the effects in the three cases from line~\ref{line:gupdate1} to \ref{line:gupdate3}.  We will use the DFS post-order in \cref{alg:lazygUpdate} where the base case consist of the boundary of visited nodes (the first two cases in line~\ref{line:gupdate1} and \ref{line:gupdate2}).  If a node $v\in U_1(E)$ satisfies line~\ref{line:gupdate1}, we have
\begin{align*}
    \val'(v) =& S\otimes \val(v) = S\otimes\left(\left(\prod_{i = 0}^k \pend(u_i)\right)\otimes  \val(v)\right)\tag{by \cref{lem:lazyg3}}\\
    =& S\otimes W^{(t-1)}(N(v))\tag{induction hypothesis}\\
    =& S\otimes\left(\oplus_{x\in N(v)} W^{(t-1)}(x)\right)\tag{definition of $W^{(t-1)}(N(v))$}\\
    =& \oplus_{x\in N(v)} S\otimes W^{(t-1)}(x)\tag{distributive property}\\
    =& W^{(t)}(N(v)).
\end{align*}
The second case $v\in U_2(E)$ is trivial.   For the third case in line~\ref{line:gupdate3}, $v$ will not be a leaf node and $N(u)$ for all $u\in \mathcal{C}(v)$ forms a partition of $N(v)$.  Thus, by induction hypothesis for all $u\in \mathcal{C}(v)$ we have $\val'(u) = W^{(t)}(N(u))$, and $\val'(v) = \oplus_{u\in \mathcal{C}(v)} W^{(t)}(N(v))$ by the commutative law of $\oplus$.  Finally, by \cref{lem:lazyg3}, $\left(\prod_{i = 0}^k \pend(u_i)\right)\val'(v) = \val'(v) = W^{(t)}(N(v))$ which completes the proof. 
\end{proof}

With the above three claims, the answer of \cref{alg:lazygquery} with $E$ is 
\begin{align*}
    \oplus_{v\in U_1(E)}\val(v)
    =& \oplus_{v\in U_1(E)}{W^{(t)}(N(v))}\tag{by \cref{lem:lazyg1,lem:lazyg3}}\\
    =& W^{(t)}(E)\tag{by \cref{lem:lazyg2} and $\oplus$ commutative}
\end{align*}
which proves the correctness. 
\end{proof}

\section{Hardness to computing cost function and trading function}\label{app:hardness}
Here, we provide a cost function and a trading function that is NP-hard to compute. 

Consider the following convex function on non-negative integers $$C_{\text{partition}}(\vw) = \mathbf{1}\left[\exists S\subset [n]\text{ so that }2\sum_{i\in S} w_i = \sum_{j = 1}^n w_j\right]+100C_Q(\vw)$$
where $C_Q$ is the QMSR in \cref{eq:QMSR}.  Note that computing $C_{\text{partition}}$ suffices to solve the partition problem which is NP-hard.  Now, we show the function is a valid cost function.  First, the function is convex, as the first term is bounded by $1$.  Moreover, the first term decides the partition problem, so adding a constant to all coordinates does not change the partition problem, and $C_{\text{partition}}(\vw+\alpha \mathbf{1}) = C_{\text{partition}}(\vw)+\alpha$ which is $\mathbf{1}$-invariant.  Finally, we can extend the domain to $\R^n$ to a differential able function.  

Moreover, computing trading function can also be NP-hard, as $-C_{\text{partition}}(-\vw)$ is an concave increasing trading function.

\section{Proof of Proposition~\ref{thm:equiv_swap}}
  We first show a reduction from a market maker to a RU algorithm.  Given an initial state (the numbers of outstanding securities) $\vw_0$ on $(\mathcal{X}, \mathcal{F})$, we run RU on initial weight $\vw_0$.  In each round with any $E^-, E^+, s_+$, because $\varphi$ is increasing and there exists $s\le \lambda$ that solves $\varphi(\vw+s_+\mathbf{1}_{E^+}-s\mathbf{1}_{E^-}) = \varphi(\vw)$, $\varphi(\vw)$ is between $\varphi(\vw+s_+\mathbf{1}_{E^+}-\lambda \mathbf{1}_{E^-})$ and and $\varphi(\vw+s_+\mathbf{1}_{E^+})$.  Then we can use binary search to find $s$ where each iteration require one update and one query.  The backward trade follows similarly.

    For the other direction, given a swap market maker, we construct RU with the following reduction.  Given an initial weight $W_0$, we pick one element $x^*$ so that $\{x^*\}\in \mathcal{F}$ and create additional variables $\phi, M, M'$ with initial value equal to $\varphi(w_0), w_0(x^*)$ and $w_0(x^*)$ respectively and run the swap market maker with initial state $W_0$. 
    
    For each update with $E$ and $S$, if $S \ge 0$ we run the forward trade operation, $s^- = \tradef(\{x^*\}, E, S)$ from the market maker, and update $M\gets M+S\mathbf{1}[x^*\in E]$ and $M' \gets M'+S\mathbf{1}[x^*\in E]-s^-$. If $S<0$, we run the backward trade operation, $s^+ = \tradef( E,\{x^*\}, S)$ from the market maker, and update $M\gets M+S\mathbf{1}[x^*\in E]$ and $M' \gets M'+S\mathbf{1}[x^*\in E]+s^+$.  Then we compute $\varphi(W)$ from $M, M'$, and $\phi$.  We can see the market maker's state $w = (W_{-x^*}, W_{x^*}')$ maintains the range query problem $W$ excepts for coordinate $x^*$, and store $M = W_{x^*}$ and $M' = W_{x^*}'$.  Therefore, we can recover the value $\varphi(W)$ from $M, M'$ and $\phi = \varphi(w)$.

\end{document}